\documentclass[english,reprint,aps,superscriptaddress]{revtex4-2}
\usepackage[T1]{fontenc}
\usepackage[utf8]{inputenc}
\usepackage{babel}
\usepackage{mathrsfs}
\usepackage{bm}
\usepackage{amsmath}
\usepackage{amsthm}
\usepackage{amssymb}
\usepackage{graphicx}
\usepackage[unicode=true,bookmarks=false,breaklinks=false,pdfborder={0 0 1},colorlinks=false]{hyperref}
\hypersetup{colorlinks,linkcolor=blue,citecolor=blue}
\makeatletter

\usepackage[normalem]{ulem}
\usepackage{babel}
\usepackage{verbatim}
\usepackage{amsthm}
\usepackage{xargs}[2008/03/08]
\usepackage{sidecap}

\newif\ifarxiv
\arxivtrue

\let\mainsection\section
\usepackage{color}% 
\RequirePackage[dvipsnames,usenames]{xcolor}
\usepackage{newtxtext}

\newcommand\smref[1]{SM-\ref*{#1}~\cite{sm_ref}}
\newcommand\smrefnocite[1]{SM-\ref*{#1}}

\ifarxiv
 \let\savedaddcontentsline\addcontentsline 
\global\long\def\addcontentsline#1#2#3{}% 

\else

\usepackage{xr}

\newcommand*{\addFileDependency}[1]{% 
\typeout{(#1)}% 
\@addtofilelist{#1}
\IfFileExists{#1}{}{\typeout{No file #1.}}
}
\newcommand*{\myexternaldocument}[1]{% 
\externaldocument{#1}% 
\addFileDependency{#1.tex}% 
\addFileDependency{#1.aux}% 
}

\myexternaldocument{sm}

\fi

\definecolor{DarkBlue}{rgb}{0.0, 0.0, .7}

\newcommand\ns[1]{#1}
\newcommand\cut[1]{}

\global\long\def\papertitle{Optimization of nonequilibrium free energy harvesting illustrated on bacteriorhodopsin} 

\newcommand{\feasibleset}{\Lambda}

\global\long\def\numstates{n}% 
\global\long\def\RB{R^b}% 
\global\long\def\RC{R^c}% 
\global\long\def\RBC{R}% 

\global\long\def\pp{\bm{p}}% 
\global\long\def\pstB{{\bm{\pi}^b}}% 
\global\long\def\pstBnb{\pi^b}% 
\global\long\def\pstBCnb{\pi}% 
\global\long\def\pstBC{\bm{\pi}}% 
\global\long\def\pwr{\Theta}% 
\global\long\def\pwrSt{\theta}
\global\long\def\invBeta{{\beta}^{-1}}% 

\global\long\def\JJ{{\boldsymbol{J}}}
\global\long\def\JJji{J_{ji}}
\global\long\def\JJij{J_{ij}}

\global\long\def\JJcij{{J^c_{ij}}}
\global\long\def\JJcji{{J^c_{ji}}}
\global\long\def\JJc{{\boldsymbol{J}}^c}

\global\long\def\incmatrix{\mathbb{B}}
\global\long\def\epr{\dot{\Sigma}}

\global\long\def\FF{\mathcal{F}}% 
\global\long\def\GG{\mathcal{G}}% 
\newcommand\GGd{\dot{\GG}}% 
\newcommand\FFd{\dot{\FF}}% 

\newcommand\FFdB[1]{\FFd^b(#1)}% 
\newcommand\FFdC[1]{\FFd^c(#1)}% 
\newcommand\GGdBbase{\GGd^b}% 
\newcommand\GGdCbase{\GGd^c}% 
\newcommand\GGdB[1]{\GGdBbase(#1)}% 
\newcommand\GGdC[1]{\GGdCbase(#1)}% 
\newcommand\GGdBCbase{\GGd^\mathrm{tot}}% 

\newcommand\derivSRB[1]{\dot{S}^b(#1)}

\global\long\def\gC{\bm{g}^c}% 
\global\long\def\gCi{{g}^c}% 
\global\long\def\gB{\bm{g}^b}% 
\global\long\def\gBi{{g}^b}% 
\global\long\def\ggB{\dot{\bm{g}}^b}% 
\global\long\def\ggBi{\dot{g}^b}% 

\global\long\def\psiLRvec{\boldsymbol{\phi}^{\text{LR}}}
\global\long\def\hvecbold{\boldsymbol{\phi}}
\global\long\def\Gmax{\mathscr{G}}% 
\global\long\def\hvec{\phi}
\global\long\def\psiM{\phi^{\text{M}}}% 
\global\long\def\psiLR{\phi^{\text{LR}}}% 
\global\long\def\popt{\pp^{*}}% 
\global\long\def\popti{{p}^*}% 
\global\long\def\iopt{{i}^*}% 

\newcommand{\evndx}{\alpha}

\global\long\def\uu{\boldsymbol{u}}
\global\long\def\mm{\boldsymbol{m}}
\global\long\def\mma{\mm^{\evndx}}
\global\long\def\uua{\uu^{\evndx}}
\newcommand{\uuIX}[1]{\uu^{#1}}

\allowdisplaybreaks

\makeatother

\begin{document}
\title{\papertitle}
\author{Jordi Piñero}
\affiliation{ICREA-Complex Systems Lab, Universitat Pompeu Fabra, 08003 Barcelona,
Spain}
\author{Ricard Solé}
\affiliation{ICREA-Complex Systems Lab, Universitat Pompeu Fabra, 08003 Barcelona,
Spain}
\affiliation{Institut de Biologia Evolutiva (CSIC-UPF), 08003 Barcelona, Spain}
\affiliation{Santa Fe Institute, Santa Fe, New Mexico 87501, United States}
\author{Artemy Kolchinsky}
\email{artemyk@gmail.com}

\affiliation{ICREA-Complex Systems Lab, Universitat Pompeu Fabra, 08003 Barcelona,
Spain}
\affiliation{Universal Biology Institute, The University of Tokyo, 7-3-1 Hongo,
Bunkyo-ku, Tokyo 113-0033, Japan}
\begin{abstract}
Harvesting free energy from the environment is essential for the operation of many biological and artificial systems. We use techniques from stochastic thermodynamics to investigate the maximum rate of harvesting achievable by optimizing a set of reactions in a Markovian system, possibly under various kinds of 
topological, kinetic, and thermodynamic constraints. 
This question is relevant for the optimal design of new harvesting devices as well as for quantifying the efficiency of existing systems. 
We first demonstrate that the maximum harvesting rate can be expressed as a constrained convex optimization problem. We illustrate it on bacteriorhodopsin, a light-driven proton pump from Archaea, which we find is close to optimal under realistic conditions. 
In our second result, we solve the optimization problem in closed-form in three physically meaningful limiting regimes. These closed-form solutions are illustrated on two idealized models of unicyclic harvesting systems.
\end{abstract}
\maketitle

\mainsection{Introduction}
Many molecular systems, both biological and artificial, harvest free energy from their environments. Biological organisms require  free energy 
to grow and replicate~\cite{hill_free_1977,morowitz2012foundations}, and
many species undergo selection for increased  harvesting~\cite{lotka1922contribution,brown1993evolution,watt1985bioenergetics,judson2017energy}. Artificial harvesting systems have also been constructed and optimized in the field of synthetic biology~\cite{choi2005artificial,lee2018photosynthetic,berhanu2019artificial,villa2019fuel,kleineberg2020light,miller2020light,guindani2022synthetic,albanese2023light}. The optimization of free energy harvesting is thus a central problem both in biology and engineering.

As an example, consider a harvesting system such as a biological metabolic network that converts glucose to ATP~\cite{chandel2021glycolysis}. Suppose that the kinetic and thermodynamic parameters of one or more reactions can be optimized, either by natural selection or artificial design. What is the maximum rate of free energy harvesting that the network can achieve, and what are the kinetic and thermodynamic parameters that achieve it?  These questions are relevant both for design of optimal harvesting devices and for quantifying the efficiency of existing systems. % 

In this paper, we use techniques from stochastic thermodynamics to derive bounds on maximum rate of free energy harvesting.  
\ns{
We consider a harvesting system in
nonequilibrium steady state which is coupled to an external source of free
energy, an internal free energy reservoir, and a heat bath. 
The setup is well-suited for studying the kinds of small-scale systems usually considered in stochastic thermodynamics~\cite{seifert2012stochastic}, where assumptions of local detailed balance and steady state are justified. The steady-state assumption is reasonable in many molecular systems, where there is a separation of timescales between internal relaxation and environmental change. 

We suppose that the system's dynamics can be separate into two kinds of processes, termed \emph{baseline} and \emph{control}. % 
The baseline processes, which are held fixed, mediate the coupling to the external source of free energy. Control refers to all other processes which can be optimized for increasing the \emph{harvesting rate} at which free energy flows into the internal reservoir. 
The particular separation of baseline/control generally depends on domain knowledge about the system and the scientific question at hand. For example, to study the efficiency of a particular reaction in a metabolic network, that reaction may be treated as control while the other reactions are baseline.
The baseline/control separation is similar to the distinction in control theory between ``plant'' (a given system with fixed dynamical laws) and ``controller'' (the part that undergoes optimization)~\cite{doyle2013feedback}.
}

In our first set of results, we demonstrate that the optimization of the harvesting rate can be expressed as a convex optimization problem. The optimal solution of this problem determines both the maximum harvesting rate and the specific control processes  that achieve that maximum.  Importantly, the optimization can also account for various types of constraints on the possible network topology, kinetic timescales, and thermodynamic forces of the control processes. 

\begin{figure}[b]
\centering
\includegraphics[width=1\columnwidth]{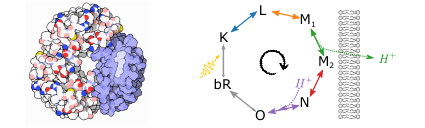}
 
\caption{\label{fig:1}Left: Bacteriorhodopsin is a biomolecular free energy harvesting machine~\cite{lanyi2004bacteriorhodopsin}, illustrated in its trimer configuration by D.~Goodsell (CC-BY-4.0)  \cite{goodsell2002bacteriorhodopsin}.  
Right: during each turn of the bacteriorhodopsin photocycle, the molecule absorbs a photon 
and pumps a proton against the cell's membrane potential. 
}
\end{figure}

We illustrate our results on bacteriorhodopsin (Fig.~\ref{fig:1}), a proton-pumping membrane protein. Bacteriorhodopsin is a photosynthetic system found in Archaea, with close relatives in bacteria~\cite{beja2000bacterial,gomez2019microbial}.
It is also
used in 
many artificial light-harvesting systems~\cite{choi2005artificial,berhanu2019artificial,lee2018photosynthetic,kleineberg2020light}. \ns{Using published thermodynamic and kinetic data, we develop a thermodynamically consistent stochastic model of bacteriorhodopsin.} 
We demonstrate that, under normal operating conditions, the bacteriorhodopsin system is remarkably efficient.

\ns{Our main result is formulated as a convex optimization problem which must be solved numerically in the general case. In the second part of
this paper, we derive closed-form solutions of this problem for
three physically meaningful regimes:
the weakly-driven linear response regime, the irreversible deterministic regime, and the intermediate near-deterministic regime. 
These solutions illustrate how optimal harvesting reflects the ``alignment'' between free energy input and relaxation dynamics. We illustrate these closed-form solutions on two unicyclic systems, which may be interpreted as idealized models of two types of nonequilibrium harvesting devices.}

\ns{We finish our paper with a brief Discussion. There we relate our approach to previous work, including maximization of power output in steady-state engines and flux balance analysis. We also propose directions for future research.}

\mainsection{Setup}
We consider a system with $n$ mesostates described by the distribution
$\pp=(p_{1},\dots,p_{n}) \in \mathbb{R}_+^n$.  The distribution evolves according to the master
equation $\dot{p}_{i}=\sum_{j}\RBC_{ij}p_{j}$, where $\RBC_{ji}$ is
the transition rate $i\to j$ ($\RBC_{ii}=-\sum_j \RBC_{ji}$). 
Usually $\pp$ represents a probability distribution of a stochastic system with Markovian dynamics~\cite{schnakenberg1976network,esposito2010three}. 
However, under an appropriate choice of units, it may also represent chemical concentrations in a deterministic chemical reaction network with pseudo-first-order reactions, such as an enzymatic cycle~\cite{wachtel2018thermodynamically,cornish2013fundamentals}.

The system is coupled to a heat bath at inverse temperature $\beta=1/k_B T$, % 
an internal free energy reservoir, and 
a nonequilibrium environment that serves as an external source of free energy. For example, in a metabolic network, one may consider an internal reservoir of ATP and an external source of glucose. 
The system has % 
nonequilibrium free energy
\begin{align}
\FF(\pp)=\sum_{i}p_{i}f_{i}-\invBeta S(\pp)\,,
\label{eq:FFdef}
\end{align}
where  $S(\pp):=-\sum_{i}p_{i}\ln p_{i}$
is the Shannon entropy and $f_{i}$ is the internal
free energy of mesostate $i$~\cite{horowitz2013imitating}. 

\cut{To investigate this question, we separate the dynamics of the system into two kinds of processes, termed \emph{baseline} and \emph{control}. % 
Control processes are those that are optimized to increase harvesting, while the baseline processes are held fixed.
The particular separation of baseline/control generally depends on domain knowledge about the system and the scientific question. For example, to study the efficiency of a particular reaction in a metabolic network, that reaction may be treated as control while the other reactions are baseline.
The baseline/control separation is similar to the distinction in control theory between ``plant'' (a given system with fixed dynamical laws) and ``controller'' (the part that undergoes optimization)~\cite{doyle2013feedback}. 
In  stochastic thermodynamics, the baseline/control distinction has been    
used to study autonomous Maxwellian demons~\cite{shiraishi2015role,shiraishi2016measurement}, 
counterdiabatic driving~\cite{takahashi2020nonadiabatic}, and the cost of maintaining a nonequilibrium steady state~\cite{horowitz2017information,horowitz_minimum_2017}. 
}

As mentioned in the Introduction, we suppose that the dynamics of the system can be separated into \emph{baseline} and \emph{control} processes. 
We make one important assumption in our analysis:
the control processes are only coupled to the heat bath and internal free energy reservoir, but not directly to the external source of free energy. 
This means that control can only increase harvesting by interacting with the baseline, rather than directly increasing the inflow of free energy from the external source. 
For example, in a metabolic network, control processes cannot directly increase the import of glucose, but they can affect the rate at which glucose is converted into ATP by optimizing intermediate reactions and mechanisms. Control processes may be driven by the internal reservoir (e.g., coupled to ATP hydrolysis) or undriven (e.g.,  enzymes).

To formalize the baseline/control distinction, we write the rate matrix as $\RBC=\RB+\RC$, where $\RB_{ji}$ and $\RC_{ji}$  represent the transition rate of $i\to j$ due to baseline and control. 
Given distribution $\pp$, the increase of system free energy due to baseline processes is 
\begin{align}
\FFdB{\pp} =\sum_{i,j}p_i\RB_{ji} (f_{j}+\invBeta\ln p_{j})\,.
\label{eq:deltaF}
\end{align}
The increase due to control processes is defined analogously but using rate matrix $\RC$,
\begin{align}
\FFdC{\pp} =\sum_{i,j}p_i\RC_{ji} (f_{j}+\invBeta\ln p_{j})\,.
\label{eq:Fcdef}
\end{align}
The \emph{harvesting rate} is the rate at which free energy flows to the internal reservoir. The harvesting rate due to baseline processes is
\begin{align}
    \GGdB{\pp} &= \sum_{i,j}p_{i}\RB_{ji} \gBi_{ji}+\sum_i p_i \ggBi_{i}\,,
\end{align}
where $\gBi_{ji}$ is the free energy increase in the internal reservoir due to a single baseline transition $i\to j$ 
and $\ggBi_{i}$
is the rate of free energy flow to the {internal} reservoir due to internal transitions
within $i$ (assuming $i$ is a coarse-grained mesostate). 
Similarly, the harvesting rate due to control processes is
\begin{align}
        \GGdC{\pp} &= \sum_{i,j}p_i\RC_{ji} \gCi_{ji}\,,
\end{align}
where $\gCi_{ji}$ is the free energy increase in the internal reservoir due to  control transition $i\to j$. For simplicity, we assume that control cannot exchange free energy with the internal reservoir due to internal transitions within $i$.  Negative values of $\gBi_{ji},\ggBi_i,\gCi_{ji}$ indicate driving done on the system by the internal reservoir. 

For a concrete example of how $(\RB,\RC,\bm{f},\gB,\ggB,\gC)$ are defined for a real biomolecular system, see the bacteriorhodopsin example below and \smref{sec:br}.

As standard in stochastic thermodynamics, we assume that % 
control processes obey local detailed balance (LDB)~\cite{seifert2012stochastic,maes2021local},
\begin{align}
\ln (\RC_{ji}/\RC_{ij})= \beta( f_{i}-f_{j} - \gCi_{ji})  &\qquad\text{for } \RC_{ji} > 0\,.\label{eq:ldb}
\end{align}
Eq.~\eqref{eq:ldb} guarantees that the irreversibility of each control transition is determined by the amount of free energy dissipated by that transition~\cite{brown_sivak_2020}. Observe that the right side accounts for free energy changes of the system ($f_i-f_j$) and the internal reservoir ($\gCi_{ji}$), but not the external source. This formalizes the assumption that control processes do not exchange free energy with the external source. 

We do not require that the baseline rate matrix obeys LDB, although it will often do so % 
for reasons of thermodynamic consistency.

\mainsection{Maximum harvesting rate}
Suppose that the combined rate matrix $\RBC=\RB+\RC$ has the steady-state distribution $\pstBC$, which satisfies $\RB\pstBC + \RC\pstBC = 0$. 
The total steady-state harvesting rate due to baseline and control is
\begin{align}
\GGdBCbase= \GGdB{\pstBC} +\GGdC{\pstBC}.
\label{eq:ggG}
\end{align}

We seek to maximize this harvesting rate by varying the parameters of the control processes $(\RC,\gC)$ while holding the baseline parameters  ($\bm{f},\RB, \gB, \ggB$) fixed. \ns{Finding this maximum would allow us to determine fundamental bounds on harvesting and to evaluate the efficiency of existing harvesting systems.}

\ns{However, $\GGdBCbase$ is not a concave function of the parameters $(\RC,\gC)$, therefore maximization of~\eqref{eq:ggG} is not a convex optimization problem and is not generally intractable. In the following, we reformulate this maximization as a convex optimization with a physically interpretable objective. This allows us to solve the optimization numerically and, at least for some special cases, also in closed form.}

To begin, we rewrite \eqref{eq:ggG} as % 
\begin{align}
\GGdBCbase= \FFdB{\pstBC} + \GGdB{\pstBC} - \epr(\JJc)/\beta .
\label{eq:gg0}
\end{align}
where we introduced the Schnackenberg formula for the entropy production rate (EPR)~\cite{schnakenberg1976network},
\begin{align}
\epr(\JJc)=\sum_{i\ne j} \JJcji \ln (\JJcji/\JJcij)\ge 0\,,
\label{eq:eprdef}
\end{align}
where  $\JJcji=\pstBCnb_i \RC_{ji} \ge 0$ is the one-way probability flux due to control transition $i \to j$. 

\ns{Eq.~\eqref{eq:gg0} has an intuitive physical interpretation: the total steady-state harvesting rate  is the rate of free energy increase in the system and internal reservoir due to baseline, minus the rate of dissipation (EPR) due to the control fluxes. The derivation of this result proceeds in two steps (see \smref{sec:derivation-of-3} for details). The first step is to show that $\epr(\JJc)=-\beta[\FFdC{\pstBC} +\beta \GGdC{\pstBC}]$, which follows by combining \eqref{eq:eprdef} with \eqref{eq:Fcdef} and \eqref{eq:ldb}. This states that the EPR due to control is proportional to the free energy loss in the system and internal reservoir due to control. The second step is to show that $\FFdB{\pstBC} +\FFdC{\pstBC}=0$, which follows because the left side is the overall derivative of the nonequilibrium free energy $\FF$, as defined in \eqref{eq:FFdef}, therefore it must  vanish in steady state. The result~\eqref{eq:gg0} then follows by combining with~\eqref{eq:ggG} and rearranging.}

\ns{Importantly, when expressed in the form~\eqref{eq:gg0}, the harvesting rate is a concave function of the steady-state distribution $\pstBC$ and the control fluxes $\JJc$ (see \smref{sec:optimization-properties-Gmax}). To find the maximum harvesting rate, we optimize \eqref{eq:gg0} with respect to  $\pstBC$ and $\JJc$.}
Note that varying $\pstBC$ and $\JJc$ is equivalent to varying the control rate matrix via $\RC_{ji} = \JJcji/\pstBCnb_i$ and control driving $\gCi_{ji}$ via \eqref{eq:ldb}. 
\ns{However, when performing this optimization, we must also ensure that $\pstBC$ is the steady-state distribution induced by the fluxes $\JJc$. This condition can be expressed as a linear constraint on $\pstBC$ and $\JJc$ via % 
$\RB\pstBC + \incmatrix \JJc = 0$.} 
Here $\JJc$ is treated as a vector in $\mathbb{R}^{n^2}$ and  $\incmatrix\in \mathbb{R}^{n\times n^2}$ 
is the incidence matrix with entries $\incmatrix_{k,ij}=\delta_{ki}-\delta_{kj}$, which guarantees $\RC\pstBC=\incmatrix \JJc$.

Combining, we arrive at the bound $\GGdBCbase \le \Gmax$, where 
\begin{align}
\Gmax&=  \sup_{(\pp, \JJ) \in \feasibleset  : \RB \pp+\incmatrix \JJ =0 } \, \FFdB{\pp} + \GGdB{\pp} - \epr(\JJ)/\beta .
\label{eq:Gmax00}
\end{align}
In this expression, $\feasibleset$ is the  feasible set of distributions $\pp$ and control fluxes $\JJ$. At a minimum, $\feasibleset$ ensures the validity of the distribution $\pp$ and the fluxes $\JJ$ via the linear constraints $\sum p_i=1$, $p_i \ge 0$, and $\JJji\ge0$. We write $\sup$ instead of $\max$ because the set of allowed fluxes is potentially unbounded. \ns{Eq.~\eqref{eq:Gmax00} implies a tradeoff between the total gain of free energy in the system and internal reservoir due to baseline (which depends only on $\pp$) and the dissipation due to control fluxes (which depends only on $\JJ$).}

Importantly, the feasible set $\feasibleset$ can 
include additional convex constraints, which % 
may introduce topological, kinetic, thermodynamic, etc. restrictions on the control processes. Topological constraints restrict which transitions can be controlled; e.g., $\JJji=0$ ensures that control does not use transition $i\to j$). 
Kinetic constraints restrict timescales of control processes, as might reflect underlying physical processes like diffusion; e.g., an upper bound on control transition rate $\RC_{ji} = \JJji/p_i \le \kappa$ can be enforced via $\JJji \le p_i\kappa$.  
Thermodynamic constraints bound the affinity~\cite{schnakenberg1976network} of control transitions; e.g., $\JJji e^{-\mathcal{A}} \le \JJij \le \JJji e^{\mathcal{A}}$ ensures that $|\ln(\JJij/\JJji)|\le \mathcal{A}$. The above examples all involve linear constraints. An example of a nonlinear, but still convex, constraint is an upper bound on the EPR incurred by control, $\epr(\JJ) \le \epr^{c}_{\max}$. % 

\ns{Eq.~\eqref{eq:Gmax00} is our first main result.} Importantly, $\Gmax$ is defined purely in terms of the thermodynamic and kinetic
properties of the baseline process, along with desired  constraints encoded in $\feasibleset$.  Thus, $\Gmax$ is the maximum steady-state harvesting rate that can be achieved by any allowed control processes, given a fixed baseline.  In addition, Eq.~\eqref{eq:Gmax00} involves the maximization of a concave objective given convex constraints. 
This is equivalent to the minimization of a convex objective, thus Eq.~\eqref{eq:Gmax00} is a convex optimization problem that can be efficiently solved  using standard numerical techniques~\cite{boyd2004convex}. 
The optimization also identifies an optimal steady-state distribution $\popt$ and control fluxes ${\JJ}^*$ that achieve the maximum harvesting rate $\Gmax$ (or come arbitrarily close to achieving it). These fix the optimal control rate matrix via $R^{c*}_{ji}={\JJij^*}/p^*_i$. Thus, our optimization specifies an upper bound on harvesting as well as the optimal control strategy that achieves this bound.

There is an important special case in which our optimization problem is simplified. Suppose that $\feasibleset$ does not enforce additional constraints on $\pp$ and $\JJ$ (more generally, we permit topological constraints if the graph of allowed transitions is connected and contains all $n$ states). Then, the objective is maximized in limit of fast control, $\JJ\to \infty$ and $\epr(\JJ)\to 0$. 
We can then write \eqref{eq:Gmax00} as an optimization over steady-state distributions:
\begin{align}
\Gmax:=\max_{\pp:\sum p_i=1,p_i\ge 0}\,\FFdB{\pp}+\GGdB{\pp}\,.\label{eq:wmax-1}
\end{align}
The optimal $\popt$ is unique as long as the baseline rate matrix is irreducible. The optimal control rate matrix is very fast ($R^{c*}\to \infty$) and obeys detailed balance for $\popt$, $R^{c*}_{ji} \popti_i= R^{c*}_{ij} \popti_j$. For details, see \smrefnocite{sec:optimization-properties-wmax} and \smref{sec:achievability}.

\mainsection{Bacteriorhodopsin}

We illustrate our results using bacteriorhodopsin, a light-driven proton pump from Archaea~\cite{lanyi2004bacteriorhodopsin}. % 

We define a thermodynamically consistent 
model of the bacteriorhodopsin cycle using published thermodynamic~\cite{varo1991thermodynamics} and kinetic~\cite{lorenz-fonfriaSpectroscopicKineticEvidence2009} data ({see \smref{sec:br}}). 
The system operates in a cyclical manner, absorbing a photon
and pumping a proton during each turn of the cycle (Fig.~\ref{fig:1}, right). 
Specifically, the transition $M_{1}\to M_{2}$ pumps a proton out of the cell. This stores free energy in the internal
reservoir (the membrane electrochemical potential),
\begin{align}
g_{M_{2} M_{1}}=-g_{M_1 M_2}=e\Delta\psi-(\ln10)\invBeta\Delta\text{pH}\,,
\label{eq:pstep}
\end{align}
where $e$ is the elementary charge constant, $\Delta\psi$ is the
membrane electrical potential, and $\Delta\text{pH}$ is the membrane pH difference. 
The other transitions in the cycle do not affect the free energy of the internal reservoir ($g_{ij}=0$ and $\dot{g}_i=0$). 

During the transition $\mathrm{bR}\to K$, the system leaves the ground state
by absorbing a photon at 580nm, thereby harvesting free energy from the external source and dissipating some heat to the environment at $T=293^{\circ}$ K.  
This transition is much faster (picoseconds) than the other transitions in the photocycle (micro- to milliseconds). % 
As commonly done in photochemistry~\cite{penocchio2021nonequilibrium}, we coarse-grain
transitions $O\to\mathrm{bR}$ and $\mathrm{bR}\to K$ into a single effective transition $O\to K$. 

To explore the performance
of bacteriorhodopsin under different conditions, we vary the membrane electrical potential $\Delta\psi$
between $-75$ and $350$ mV, while using a realistic fixed $\Delta\text{pH}=-0.6$~\cite{lanyi1978light}. 
We show the actual harvesting rate ($\GGdBCbase$ in units of $k_BT/$sec) at different potentials as a black line in Fig.~\ref{fig:br-data}~(a). At a plausible \emph{in vivo} $\Delta\psi=120$ mV \cite{lanyi1978light}, the model exhibits a steady-state current of 11.5 protons/sec, each proton carrying 6.1 $k_B T$ of free energy, corresponding to ${\GGdBCbase}\approx 70 \,k_B T$/sec. The largest output is achieved near the \emph{in vivo} potential: at lower potentials, the cycle current saturates while the free energy per proton drops, and at higher potentials the pump stalls.

Next, 
we quantify the maximum harvesting rate that can be achieved by optimizing the parameters of individual transitions. 
This analysis is relevant for understanding limits  % 
on  increasing bacteriorhodopsin output, whether via natural selection or 
biosynthetic methods~\cite{miller1990kinetic,seitz2000kinetic,wise2002optimization,hillebrecht2004directed}. 
Interestingly, such transition-level optimization may be feasible in bacteriorhodopsin, as the properties of several transitions are known to be individually controlled by particular amino acid residues in the bacteriorhodopsin protein~\cite{miller1990kinetic,tittor1994specific,balashov1999proton,li2000protein}.

For each reversible transition in the cycle, for instance  $N\leftrightarrow O$, we  define the baseline as the rest of the cycle without that transition. We then optimize control under the topological constraint that only the relevant transition (e.g., $N\leftrightarrow O$) is allowed. 
The analysis is repeated for all transitions, except for the (coarse-grained) photon-absorbing transition $O\leftrightarrow K$, which is in accordance with our assumption that control cannot directly exchange free energy directly with the external source.

\begin{figure}
\includegraphics[width=1\columnwidth]{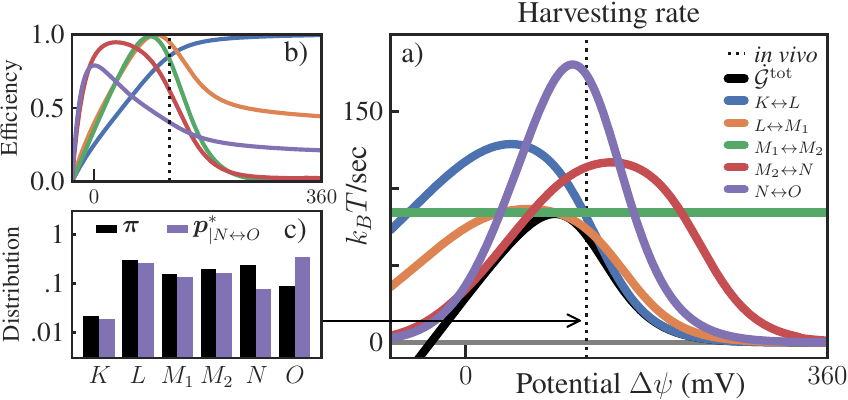}
\caption{\label{fig:br-data} 
(a) Comparison of the actual harvesting rate $\GGdBCbase$ at different electrical potentials $\Delta\psi$, versus maximum rate $\Gmax$ achieved by optimizing five intermediate transitions (color scheme from Fig.~\ref{fig:1} Right). (b) Efficiency $\GGdBCbase/\Gmax$ computed while separately optimizing each transition, with colors as in (a).  (c) The actual  steady state $\pstBC$ versus the optimal distribution $\protect\popt$ when optimizing the $N\leftrightarrow O$ transition (at $\Delta\psi=120$ mV).% 
}
\end{figure}

Fig.~\ref{fig:br-data}~(a) shows $\Gmax$, the maximum $\GGdBCbase$ achievable by optimizing each reversible transition.  In Fig.~\ref{fig:br-data}~(b),  we also show the \emph{efficiency}  $\GGdBCbase/\Gmax \le 1$ for each transition,  that is the ratio of the actual and maximum harvesting rate.

Several transitions, such as $K\leftrightarrow L$,$L\leftrightarrow M_1$,$M_1\leftrightarrow M_2$, are remarkably efficient ($\ge 85\%$) near \emph{in vivo} membrane potentials. The reprotonation step $N\leftrightarrow O$ is the least efficient ($\sim 40\%$) and also has the slowest kinetics of the 5 transitions studied in Fig.~\ref{fig:br-data}. This suggests that $N\leftrightarrow O$ is a bottleneck whose optimization can have a big impact on the harvesting rate, while optimization of other non-bottleneck transitions has a more limited effect.

Observe that $\Gmax$ for $M_1 \leftrightarrow M_2$ does not depend on $\Delta \psi$. This is because $\Gmax$ is a function of baseline properties, which do not depend on the membrane potential when $M_1 \leftrightarrow M_2$ is treated as control. Conversely, $M_1 \leftrightarrow M_2$ as control transition can be optimized by 
varying the membrane potential % 
and/or scaling up the forward/backward rates. Our results show that this transition is very close to optimal at \emph{in vivo} membrane potentials and kinetic timescales.

Optimal distributions $\popt$ are also obtained, with a typical one shown in Fig.~\ref{fig:br-data}~(c). We find a consistent shift toward state $O$, which accelerates the reset of the cycle and increases the flux across the photon-absorbing transition $O\to K$.

In~\smref{sec:br}, we illustrate how the efficiency of bacteriorhodopsin transitions can be evaluated under other types of constraints, including constraints on thermodynamic affinity, dynamical activity, and kinetics.

\begin{figure}
\includegraphics[width=\columnwidth]{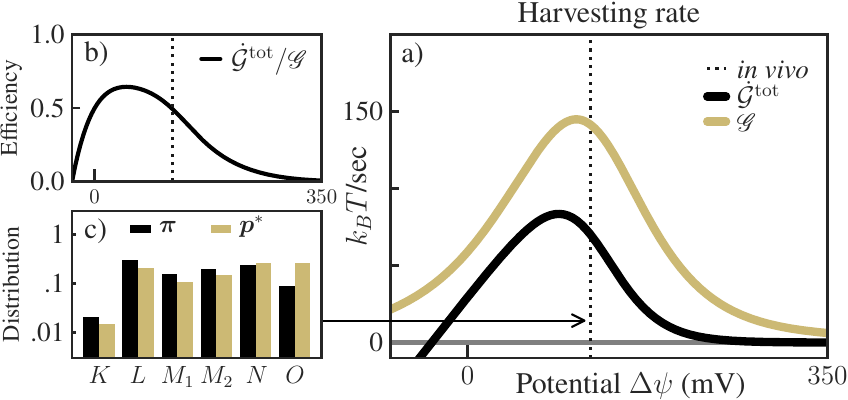}
\caption{\label{fig:br-sm-baseline-all-data}
(a) Comparison of the actual harvesting rate $\GGdBCbase$ at different electrical potentials $\Delta\psi$, versus maximum rate $\Gmax$ achieved by fixing the bacteriorhodopsin cycle as baseline and allowing any additional transitions as control. % 
(b) Efficiency of the actual bacteriorhodopsin cycle with respect to the optimized~cycle. (c) The actual steady state $\pstBC$ and optimal distribution $\pp^*$ (at $\Delta\psi=120$ mV).  % 
}
\end{figure}

As a final analysis, instead of  optimizing individual existing transitions in the bacteriorhodopsin cycle, we ask to what extent the harvesting rate can be increased by any \emph{additional} control processes. For example, this could involve an additional enzyme that shifts the cycle's steady state by permitting a new transition between distant states (e.g. $L\leftrightarrow N$), possibly yielding an increase of the proton pumping rate.

In this case, we treat the entire bacteriorhodopsin system as the baseline, and we  do not introduce any additional constraints on the steady-state distribution or the control fluxes. % 
Then, as shown in Sec.~\smref{sec:achievability}, the objective in~\eqref{eq:Gmax00} is achieved in the limit of fast control, and the maximum harvesting rate can be found by solving the simplified optimization problem~\eqref{eq:wmax-1}. 

For this setup, Fig.~\ref{fig:br-sm-baseline-all-data}~(a) shows the baseline (actual) harvesting rate $\GGdBCbase$ and the maximum harvesting rate $\Gmax$ at varying $\Delta \psi$. Interestingly, both peak at around the \emph{in vivo} values of the membrane potential. % 
In Fig.~\ref{fig:br-sm-baseline-all-data}~(b), we show that the actual bacteriorhodopsin cycle harvests approximately 50\% of the fundamental bound given by $\Gmax$  (at {\em in vivo} values of $\Delta \psi$). This suggests that bacteriorhodopsin is remarkably close to optimal, relative to improvements that could be achieved by introducing any additional control processes.

We also show the actual steady-state distribution and the optimal distribution $\popt$ in Fig.~\ref{fig:br-sm-baseline-all-data}~(c). 
The optimal distribution increases the probability of state $O$, similar to the optimal distribution found by optimizing the $N\leftrightarrow O$ transition, shown in Fig.~\ref{fig:br-data}~(c). However, unlike Fig.~\ref{fig:br-data}~(c), where most of the extra probability is taken from state $N$, in Fig.~\ref{fig:br-sm-baseline-all-data}~(c) the probability is drawn more uniformly from other states in the cycle, indicating the presence of distributed control.

\mainsection{Limiting regimes}

Our results are stated via an optimization problem that generally does not have a closed-form solution. 
In our
second set of results, we identify closed-form expressions in three physically meaningful regimes. For simplicity, here we focus on the simplified objective~\eqref{eq:wmax-1}. 
\ns{See \smref{sec:static-env} for detailed derivations, including analysis of the conditions under which each of these three approximation are be valid.}

For convenience, we  first 
rewrite~\eqref{eq:wmax-1} as % 
\begin{align}
\Gmax =\max_{\pp}\,-\derivSRB{\pp}/\beta+\sum_{i}p_{i}\hvec_{i} \,,\label{eq:wmax-lin-split}
\end{align}
{where $\derivSRB{\pp}=-\sum_{i,j}\RB_{ij}p_{j}\ln p_{i}$ is the increase of the Shannon entropy of $\pp$ under $\RB$} and for convenience we defined 
 $\hvec_{i}:=\ggBi_{i}+\sum_{j}\RB_{ji}(f_{j}-f_{i}+\gBi_{ji})$. 
The objective~\eqref{eq:wmax-lin-split} contains a nonlinear term $-\derivSRB{\pp}/\beta$ quantifying the decrease of information-theoretic entropy
and a linear term $\sum_i p_i \phi_i$ quantifying the flow of thermodynamic free energy.

Next, we consider three regimes. 

\newcommand\vecAi{u^a}
\footnotetext[200]{The dynamical activity refers to the overall rate of back-and-forth jumps across a transitions, $\protect \RB_{ij}\pstBnb_j + \RB_{ji}\pstBnb_i$.}

\emph{Linear response} (LR) applies when
the optimal distribution $\popt$ is close to the steady-state distribution of the baseline rate matrix $\RB$. Suppose that $\RB$
is irreducible and has a unique steady state $\pstB$ with full support. We introduce the ``additive reversibilization''  of $\RB$,
\begin{align}
A_{ij}=(\RB_{ij}+\RB_{ji}\pstBnb_i/\pstBnb_j)/2\,.
\end{align}
The rate matrix $A$ obeys detailed balance (DB) for the steady-state distribution $\pstB$ and has the same dynamical activity~\cite{Note200} on all edges as $\RB$. $A$ may be considered as a DB version of $\RB$, and it is equal to $\RB$ when the latter obeys DB~\cite{kolchinsky2023thermodynamic,fillEigenvalueBoundsConvergence1991}. Let $\uua$
indicate the $\evndx$-th right eigenvector of $A$ normalized as $\sum_i ({u^{\alpha}_i})^2/\pstBnb_i=1$,
and $\lambda_{\evndx}$ the corresponding real-valued eigenvalue ($\lambda_{1}=0$).   
The LR solution for the maximum harvesting rate and the optimal distribution is 
\begin{align}
\begin{aligned}
\Gmax&\approx \GGdB{\pstB}  + \beta \sum_{\evndx>1}\frac{\left|\Omega_\evndx\right|^{2}}{-\lambda_{\evndx}}\\
\popt&\approx  \quad \pstB\;\;\;+\beta\sum_{\evndx>1}\frac{\Omega_\evndx}{-\lambda_{\evndx}}\uua 
\end{aligned}
\label{eq:linear-resp-solution-2}
\end{align}
where $\Omega_\evndx = (\hvecbold +\invBeta {\RB}^T\ln\pstB)^T\uua/2$ quantifies the harvesting ``amplitude'' for mode $\evndx$. % 

Eq.~\eqref{eq:linear-resp-solution-2} has a simple interpretation. In addition to the baseline harvesting rate $ \GGdB{\pstB} $, $\Gmax$ contains
contributions from the relaxation modes of $A$, with each mode weighed by its squared amplitude and
relaxation timescale $-1/\lambda_{\evndx}$. 
All else being equal, $\Gmax$ is large when slow modes have large harvesting amplitudes. 
The optimal $\popt$ shifts the baseline steady state $\pstB$ toward mode $\evndx$ in proportion to
that mode's harvesting amplitude and relaxation
timescale, thereby optimally balancing the tradeoff between harvesting and dissipation.

The \emph{Deterministic} (D) regime applies when the nonlinear information-theoretic term in~\eqref{eq:wmax-lin-split} is much smaller than the linear thermodynamic term. 
We can then ignore the former, 
turning~\eqref{eq:wmax-lin-split} into a simple linear optimization. This gives
the approximate solution
\begin{align}
\Gmax\approx \hvec_{\iopt}\quad\qquad \popti_{i} \approx \delta_{\iopt i}\,
\label{eq:DF}
\end{align}
where $\iopt =\arg\max_{i}\hvec_{i}$ is the optimal mesostate. % 
This solution concentrates probability on a single mesostate, effectively ignoring the cost of maintaining this low-entropy distribution. 

The \emph{Near-Deterministic} (ND) regime lies between Linear Response and Deterministic ones.  
By perturbing $\popt$ around $\delta_{\iopt i}$, we can
decouple the
values of $p_{i}$ in the objective function~\eqref{eq:wmax-1}. The maximal harvesting rate and optimal distribution in this regime are then given by
\begin{align}
\begin{aligned}
    \Gmax &\approx \hvec_{\iopt }+\invBeta \sum_{i\ne \iopt }\RB_{i\iopt }(\ln \popti_{i}-1)\,,\\
\popti_{i}  &\approx \begin{cases} {\RB_{i\iopt }}/[{\beta(\hvec_{\iopt }-\hvec_{i})+\RB_{\iopt \iopt }}] &\qquad i\ne \iopt\\
 1-\sum_{i:i\ne \iopt }{\popti_{i}} &\qquad i=\iopt \end{cases}
\end{aligned}
\label{eq:wdPD}
\end{align}

The ND solution also has a simple
interpretation. It perturbs the Deterministic solution % 
by shifting probability towards states with high transition rates away from the optimal state 
(large $\RB_{i\iopt }$) and small decreases in harvesting ($\hvec_{\iopt }-\hvec_{i}$).
This balances the benefit of harvesting against the cost of pumping probability against $\RB_{i\iopt }$.

\begin{figure}[t]
\includegraphics[width=0.8\columnwidth]{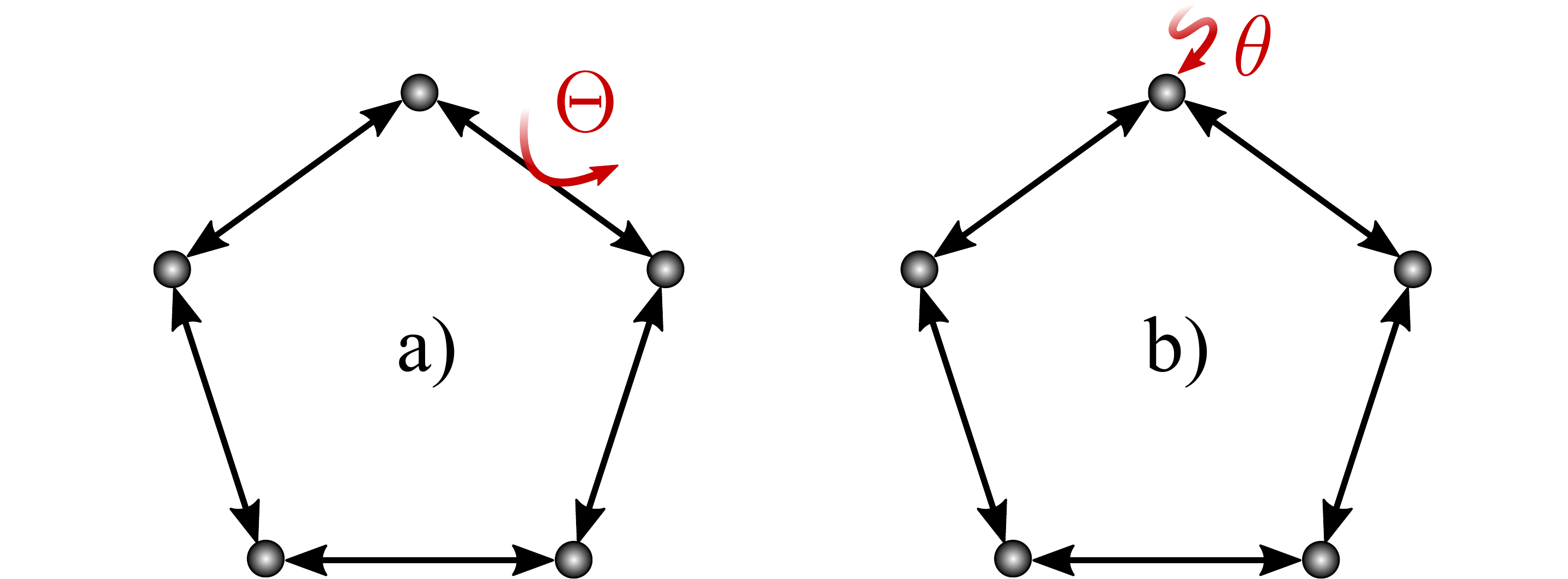}
\caption{\label{fig:two-rings} (a) Unicyclic system where free energy $\pwr$ is harvested by a single transition. (b) Unicyclic system where free energy per unit time $\pwrSt$ is harvested when the system is in a particular mesostate.
}
\end{figure}

\section{Example: unicyclic systems}

We illustrate our closed-form solutions using two simple models, both based on a unicyclic system with $n$ states. The baseline dynamics involve diffusion across a 1-dimensional
ring, with left and right jump rates set to 1. The baseline steady state is a uniform distribution, $\pstBnb_i = 1/n$, with no cyclic current. 
We assume a uniform free energy function, $f_i=0$ for all $i$.

We consider two different scenarios. 
In the first scenario,  shown schematically in Fig.~\ref{fig:two-rings}~(a), $\pwr$ of free energy is harvested each time the system carries out the transition $1\to2$, so
$$\gBi_{21} = -\gBi_{12}= \pwr\,,$$
and $\gBi_{ij}= \ggBi_i=0$ otherwise. 
This scenario may be interpreted as an idealized model of a biomolecular harvesting cycle, such as a transporter. In the second scenario, shown schematically in Fig.~\ref{fig:two-rings}~(b), free energy is harvested at a rate of $\pwrSt$ per unit time when the system is located in one particular mesostate $\iopt=1$, so
\begin{align}
 \ggBi_{1} =\pwrSt\,,
\end{align}
and  $\gBi_{ij}= \ggBi_i=0$ otherwise. This scenario may be interpreted as an idealized model of error correction or self-assembly, where free energy can only be harvested when the system is in some particular functional mesostate. 

For both scenarios, we evaluate the maximum harvesting rate $\Gmax$ and the 
optimal distributions achievable by adding any possible control to the system, without constraints. We report exact values found by numerical optimization of Eq.~\eqref{eq:wmax-lin-split}, as well as the LR, ND, and D approximations described above. To calculate the LR values, we exploit the fact that the baseline unicyclic rate matrix is a circulant matrix with a simple eigendecomposition~\cite{gray2006toeplitz}.  Full details of the derivations for the two scenarios are provided in  \smref{sec:rings1} and \smref{sec:rings2}, respectively.

\begin{figure}[t]
\includegraphics[width=1\columnwidth]{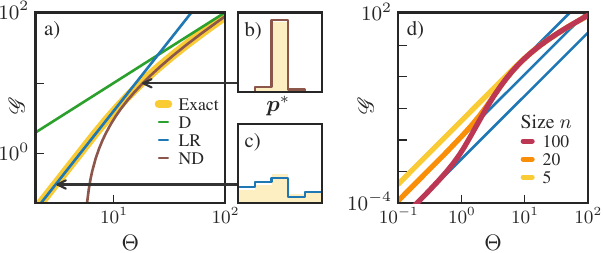}
\caption{\label{fig:main-ring} (a) Maximum harvesting rate $\Gmax$ for the unicyclic system from Fig.~\ref{fig:two-rings}~(a), as a function of supplied free energy $\pwr$. Exact value is found numerically, LR, D, and ND are calculated using approximations described in the text.
Exact and approximate optimal distributions $\popt$ in
ND (b) and LR (c) regimes are shown, with the optimal state $\iopt=1$ located in the middle of the histograms. 
(d) $\protect\Gmax$ and its LR approximation for different $\protect\pwr$ and system sizes $n$. 
}
\end{figure}
\begin{figure}[t]
\includegraphics[width=1\columnwidth]{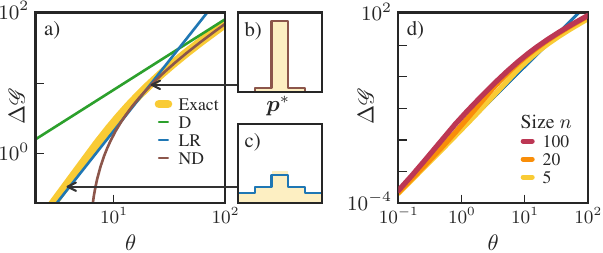}
\caption{\label{fig:state-ring} (a) Maximum harvesting rate $\Gmax$ for the unicyclic system from Fig.~\ref{fig:two-rings}~(b), as a function of supplied free energy rate $\pwrSt$. Exact value is found numerically, LR, D, and ND are calculated using approximations described in the text.
Exact and approximate optimal distributions $\popt$ in
ND (b) and LR (c) regimes are shown, with the optimal state $\iopt=1$ located in the middle of the histograms. 
(d) $\protect\Gmax$ and its LR approximation for different $\protect\pwrSt$ and system sizes $n$. 
}
\end{figure}

We first report results for the first scenario from Fig.~\ref{fig:two-rings}~(a), where  free energy is harvested during the transition $2\to 1$. 
Observe that baseline harvesting rate vanishes, $\GGdB{\pstB} =0$, since harvesting free energy requires a cyclic current. 
In Fig.~\ref{fig:main-ring}~(a), we plot the maximum harvesting rate $\Gmax$ and its approximations    as a function of the supplied free energy $\pwr$.  
For small $\pwr$, LR applies and the maximum harvesting rate is
\begin{align}
\Gmax \approx \beta \pwr^{2}{(n-1)}/{4n^{2}}.
\label{eq:lrcyc}
\end{align}
The optimal distribution in the LR regime, shown in Fig.~\ref{fig:main-ring}~(c), builds up in equal increments starting from $i=\iopt+1$ until the optimal state $\iopt=1$, after which it drops sharply. 
For large $\pwr$, 
the D regime is relevant and the
optimal distribution concentrates on the optimal state $\iopt =1$, so 
\begin{align}
    \Gmax \approx \pwr \,.
\end{align} 
At intermediate
$\pwr$, the ND regime applies, which gives
\begin{align}
\Gmax \approx \pwr-\invBeta\big\{2+\ln[2(\beta \pwr-1)(\beta\pwr-2)]\big\}.
\end{align}
The optimal distribution in the ND regime, shown in Fig.~\ref{fig:main-ring}~(b), allocates $\popti_{\iopt-1} = 1/(\beta \pwr -2)$, $\popti_{\iopt+1}=1/(2\beta \pwr - 2)$ and the rest to the optimal state $\popti_{\iopt}$.

Next, we consider the second scenario from Fig.~\ref{fig:two-rings}~(b), where free energy is harvested when the system is in the optimal mesostate $\iopt=1$.  Observe that the uniform baseline steady state assigns $1/n$ probability to the optimal state, thus in this scenario the baseline harvesting rate is $\GGdB{\pstB}=\pwrSt/n$.  To facilitate comparison with the first scenario, we focus on the {increase} of the maximum harvesting rate relative to baseline,
\begin{align}
    \Delta \Gmax := \Gmax  - \GGdB{\pstB}  = \Gmax  - \pwrSt/n . 
\end{align}
In Fig.~\ref{fig:state-ring}~(a), we show $\Delta \Gmax$ and its approximations as a function of the free energy supply rate $\pwrSt$. 
For small $\pwrSt$, LR applies and the maximum harvesting rate is
\begin{align}
\Delta \Gmax \approx  \beta \pwrSt^2 /48 \,.
\end{align}
The optimal distribution in the LR regime, shown in Fig.~\ref{fig:state-ring}~(c),   is symmetric about the optimal state $\iopt=1$. 
For large $\pwrSt$, 
the D regime is relevant and the
optimal distribution concentrates on the optimal state $\iopt =1$, so 
\begin{align}
    \Delta \Gmax \approx \pwrSt -\pwrSt/n \,.
\end{align} 
At intermediate
$\pwrSt$, the ND regime applies, giving
\begin{align}
\Delta \Gmax \approx \pwrSt-\pwrSt/n-2\invBeta \left[1+\ln\left(\beta \pwrSt-2\right)\right].
\end{align}
The optimal distribution in the ND regime, shown in Fig.~\ref{fig:state-ring}~(b),  allocates $\popti_{\iopt - 1}=\popti_{\iopt + 1} = 1/(\beta \pwrSt - 2)$ and the rest to $\popti_{\iopt}$.

There are some similarities among the two harvesting scenarios. For both scenarios, in the LR regime, the increase of the harvesting rate scales quadratically in the supplied free energy ($\pwr$ or $\pwrSt$) and linearly in inverse temperature $\beta$. 
This scaling reflects the fact that the optimal strategy has to balance  harvesting ($\pwr$ or $\pwrSt$ contributions) with the thermodynamic cost of maintaining a low entropy $\popt$ ($\beta$ contributions). 
In the ND and D regimes, $\Gmax$ scales linearly in the supplied free energy and loses its linear dependence on $\beta$.
Thus, outside of LR, the thermodynamic cost of maintaining a low entropy distribution has a minor effect on the optimal strategy. 

There are also important differences between the two scenarios. For the first scenario, the optimal strategy maintains an asymmetric $\popt$, thereby generating a net flux across the transition $2 \to 1$.  In the LR regime, the cost of maintaining this asymmetric distribution grows with the system size $n$, therefore the maximum harvesting rate in Eq.~\eqref{eq:lrcyc} scales as $\sim O(n^{-1})$. This is shown in Fig.~\ref{fig:main-ring}~(d), where we plot $\Gmax$ and its LR approximation at various $\pwr$ and $n$.  For the second scenario, the optimal strategy maintains a peaked but symmetric $\popt$. Remarkably, the cost of maintaining this distribution does not depend on system size $n$. This is shown in Fig.~\ref{fig:state-ring}~(d), where we plot $\Delta\Gmax$ at various $\pwrSt$ and $n$.

\mainsection{Discussion}

In this paper, we consider the problem of optimizing free energy harvesting in a nonequilibrium steady-state system. We demonstrate that this problem can be formulated as a constrained convex optimization problem, and we use this formulation to study optimal harvesting and efficiency  in the bacteriorhodopsin proton pump. We also solve the convex optimization problem in closed-form for three limiting regimes, as illustrated on two unicyclic models discussed above.

A key step in our analysis is to separate the dynamics of the system into separate contributions from fixed baseline processes and optimizable control processes.  We note that, in stochastic thermodynamics, the baseline/control separation has been previously    
used to study autonomous Maxwellian demons~\cite{shiraishi2015role,shiraishi2016measurement}, 
counterdiabatic driving~\cite{takahashi2020nonadiabatic}, and the cost of maintaining a nonequilibrium steady state~\cite{horowitz2017information,horowitz_minimum_2017}.

We derive a simplified bound on the maximum harvesting rate in~\eqref{eq:wmax-lin-split}, which is achieved in the limit of fast dissipation-less control. Interestingly, this expression involves a tradeoff between two terms, one information-theoretic and one thermodynamic. % 
At first glance, this % 
resembles information/free-energy tradeoffs characteristic of Maxwellian
demons and other ``information engines''
\citep{parrondo2015thermodynamics,sagawa2009minimal,abreu_thermodynamics_2012,cao_thermodynamics_2009,ito2013information,hartich_stochastic_2014,barato_efficiency_2014,sartori2014thermodynamic}. 
However, there are important differences. In a typical information engine, there is no external source of driving and information serves as \emph{fuel}, which
can be converted into $\invBeta \ln 2$ of thermodynamic free energy per bit. 
In our case, there is an external source of free energy that in some cases can be harvested more effectively by reducing the system's statistical entropy, e.g. by concentrating it on favorable states.
Here, a bit of information can increase the harvesting rate by a very large amount (much larger than $\invBeta \ln 2/\text{bit}$), and information acts more like a \emph{catalyst} than a fuel~\citep{barham1996dynamical,kolchinskySemanticInformationAutonomous2018}. Loosely speaking, this is similar to how information encoded in the sequence of a metabolic enzyme is not itself fuel, but rather allows metabolism to harvest a large amount of fuel from elsewhere. 

We finish by mentioning some connections to previous work and future directions.
First, our approach may be related to prior work on optimizing power output and free energy transduction in steady-state molecular machines~\cite{pietzonkaUniversalBoundEfficiency2016,brown2017allocating,brown2018allocating,pietzonkaUniversalTradeOffPower2018,wagoner2019mechanisms,brown_sivak_2020,schmiedl2008efficiency}. Here we consider the general problem of optimizing a set of control processes, given a fixed baseline and possible additional constraints on kinetics, topology, and thermodynamics. Previous work does not make the baseline/control distinction; instead, it is mostly concerned with the problem of optimizing system performance with respect to a small set of specific parameters or observables of interest, such as the location of free energy barriers~\cite{schmiedl2008efficiency,brown2017allocating,brown2018allocating,wagoner2019mechanisms}, efficiency~\cite{schmiedl2008efficiency}, and the size of  fluctuations~\cite{pietzonkaUniversalBoundEfficiency2016,pietzonkaUniversalTradeOffPower2018}.

There is also an interesting relation between our work and flux balance analysis (FBA)~\cite{beard2002energy,kauffman2003advances,kschischo2010gentle}. The goal of FBA is to identify deterministic fluxes in biological metabolic systems that optimize biomass production, or other similar metrics of performance. 
This can be formulated as a linear program, which may include linear constraints that enforce thermodynamically favored reaction directions~\cite{kschischo2010gentle} (interestingly, in Ref.~\cite{fleming2012variational}, the authors propose a version of FBA that also accounts for the entropy production rate). Our setting and optimization are different from FBA and its variants. We seek to optimize free energy harvesting at the stochastic level, and our objective involves nonlinear information-theoretic contributions to free energy. In addition, our optimization involves both the steady-state distribution $\pp$ and fluxes $\JJ$, which allows us to optimize harvesting due to to internal transitions within coarse-grained mesostates, as in Fig.~\ref{fig:two-rings}~(b). Nonetheless, investigating the relationship between our approach and FBA is an interesting direction for future work.

Another interesting direction for future work is to consider stochastic fluctuations of free energy harvesting. In particular, the thermodynamic uncertainty relation may be used to study tradeoffs between the entropy production rate, the average harvesting rate (the quantity $\GGdBCbase$  considered here), and the fluctuations in the amount of harvested free energy~\cite{gingrich2016dissipation,pietzonkaUniversalTradeOffPower2018}. For biomolecular systems, large fluctuations in harvesting can lead to starvation, 
suggesting that minimizing fluctuations may be of significant biological importance.

Finally, an interesting direction for future work is to consider free energy harvesting in a system embedded in a fluctuating environment. For example, one may imagine a harvesting system in an environment with fluctuating sugar sources, or with a fluctuating amount of available light. In this setting, it is  natural to optimize the harvesting rate under the topological constraint that control fluxes cannot directly change the state of the environment, for instance using bipartite models of Markovian dynamics~\cite{horowitz2014thermodynamics}. 
It would be interesting to investigate how, under the optimal harvesting strategy, the information flow % 
from the environment to the system varies with the abundance of free energy and complexity of the environment.

\vspace{10pt}

\begin{acknowledgments}

We thank members of the Complex Systems Lab, B. Corominas-Murtra, L. Seoane, D. Wolpert, and D. Sowinski for useful discussions. A.K. also thanks Sosuke Ito for support and encouragement. This project has received funding from the European Union’s Horizon 2020 research and innovation programme under the Marie Sk\l{}odowska-Curie Grant Agreement No. 101068029. J.P. was supported by the Mar\'ia de Maezt\'u fellowship MDM-2014-0370-17-2 and Grant No. 62417 from the John Templeton Foundation. The opinions expressed in this publication are those of the authors and do not necessarily reflect the views of the John Templeton Foundation.

\end{acknowledgments}

\vfill
\bibliographystyle{ieeetr}
\bibliography{main}

\clearpage

\newcommand\newstuff[1]{#1}

\ifarxiv 
\let\addcontentsline\savedaddcontentsline
% !TeX root = sm.tex

\newcommand\figurefontcmd{\fontsize{10pt}{12bp}\selectfont}

\makeatletter
\renewcommand\section{\@startsection{section}{1}{\z@}%
                                  {-3.5ex \@plus -1ex \@minus -.2ex}%
                                  {2.3ex \@plus.2ex}%
                                  {\fontsize{12pt}{12bp}\selectfont\bfseries\centering}}
\renewcommand\subsection{\@startsection{subsection}{2}{\z@}%
                                  {-3.5ex \@plus -1ex \@minus -.2ex}%
                                  {2.3ex \@plus.2ex}%
                                  {\fontsize{11.5pt}{12bp}\selectfont\bfseries\centering}}

\renewcommand\subsubsection{\@startsection{subsubsection}{3}{\z@}%
                                  {-3.5ex \@plus -1ex \@minus -.2ex}%
                                  {2.3ex \@plus.2ex}%
                                  {\fontsize{11.5pt}{12bp}\selectfont\em\centering}}
\makeatother

% \global\long\def\psiLR{\psi^{\text{LR}}}%
\global\long\def\psiM{{ \psi^{\text{M}}}}
\global\long\def\psiMmin{{\phi_{\min}}}

% \global\long\def\deltaWmaxLR{\deltaW_{\text{LR}}^{*}}%
% \global\long\def\deltaWmaxED{\deltaW_{\text{FE}}^{*}}%
% \global\long\def\deltaWmaxFE{\deltaW_{\text{FE}}^{*}}%
% \global\long\def\deltaWmaxDF{\deltaW_{\text{M}}^{*}}%

%\global\long\def\deltaW{ {\color{red}\Delta \mathscr{F}}}
\global\long\def\deltaWmaxLR{ {\mathscr{F}_{\mathrm{LR}}}}
\global\long\def\deltaWmaxDF{ {\mathscr{F}_{\mathrm{D}}}}
\global\long\def\deltaWmaxFE{ {\mathscr{F}_{\mathrm{ND}}}}
\global\long\def\deltaWmax{ {\mathscr{F}}}
\global\long\def\psiLR{ {\phi^{LR}}}
\global\long\def\psiM{ {\phi^{M}}}
\global\long\def\psiLRvec{\boldsymbol{\phi}^{\text{LR}}}
\global\long\def\hvecbold{\boldsymbol{\phi}}
\global\long\def\vv{\boldsymbol{v}}
\newcommand\rateparam{\kappa}

\global\long\def\pstBCa{\pstBC({\rateparam})}

\global\long\def\zz{\boldsymbol{z}}
\global\long\def\ppLR{\popt_{\mathrm{LR}}}
\global\long\def\ppM{\popt_{\mathrm{D}}}
\global\long\def\ppFE{\popt_{\mathrm{ND}}}
%\global\long\def\dB{ {\color{red} d}}
%\global\long\def\qB{ {\color{red} d}}
%\global\long\def\dB{ {\color{red} d}}
\sloppy
\allowdisplaybreaks
\raggedbottom

\global\long\def\gammaConst{\gamma}%

\onecolumngrid
\setcounter{equation}{0}

\renewcommand{\theequation}{S\arabic{section}.\arabic{equation}}
\renewcommand{\thefigure}{S\arabic{figure}}

% \renewcommand\thesection{\arabic{section}}
% \renewcommand\thesubsection{\arabic{subsection}}
% \renewcommand\thesubsubsection{\arabic{subsubsection}}
% \renewcommand\theequation{S\arabic{equation}}
% \makeatletter 
% %\renewcommand*{\p@subsection}{\thesection.}
% \def\p@subsection     {\thesection.}
% \def\p@subsubsection  {\thesection.\thesubsection.}
%\makeatother 

%\titleformat*{\section}{\fontsize{16}{20}\selectfont} 
%\titleformat*{\subsection}{\fontsize{13}{17}\selectfont}

% \titleformat*{\section}{\fontsize{14}{20}\sffamily\bfseries\selectfont} 
% \titleformat{\subsection}[hang]{\fontsize{12}{17}\sffamily\bfseries\selectfont}{\thesection.\thesubsection}{10pt}{}{}
% \titleformat{\subsubsection}[hang]{\fontsize{11}{15}\sffamily\itshape\selectfont}{\thesection.\thesubsection.\thesubsubsection}{10pt}{}{}
\setcounter{section}{0}

\begin{center}

{\fontsize{15pt}{22bp}\bfseries Supplemental Material:}

\vspace{5pt}
{\fontsize{15pt}{22bp}\bfseries \papertitle }

\vspace{15pt}

{\fontsize{14pt}{15} Jordi Pi\~nero, Ricard Solé, and Artemy Kolchinsky }

% \vspace{20pt}
% {\fontsize{15pt}{22bp}\bfseries Supplementary Material }

\end{center}

\fontsize{11pt}{14.5bp}\selectfont
\tableofcontents{}

\fontsize{11pt}{22bp}\selectfont 

\clearpage

\section{Derivations of main results}
\global\long\def\Lobj{\mathscr{L}}% 
\global\long\def\qq{\bm{q}}% 

% In this Section, we derive various results related to the variational principle \eqref{eq:Gmax00} suggested in the main text,
% \begin{align}
% \Gmax&=  \sup_{(\pp, \JJ) \in \feasibleset  : \incmatrix \JJ =-\RB \pp } \, \FFdB{\pp} + \GGdB{\pp} - \epr(\JJ)/\beta .
% \label{eq:appGmax00}
% \end{align}
% We note that as long as the feasible set allows for the presence of no control ($\bm{0} \in \feasibleset$), this maximum cannot be less than the harvesting rate achieved by any baseline steady state. That is, for any baseline steady state $\pstB$ that satisfies $\RB \pstB =0$, the pair $(\pstB,\bm{0})$
% \begin{align}
%     \Gmax \ge \GGdB{\pstB} 
% \end{align}
% for any $\pstB$
% \$$
% Here we note that, 
% \pstB$

\subsection{Derivation of~\eqref{eq:gg0} from LDB and steady-state assumption}
\label{sec:derivation-of-3}

Here we derive~\eqref{eq:gg0} in the main text, which reads as 
\begin{align}
\GGdBCbase= \FFdB{\pstBC} + \GGdB{\pstBC} - \epr(\JJc)/\beta .
\label{eq:appgg0}
\end{align}
To derive this result,
\begin{align}
\FFdC{\pp}+\GGdC{\pp} &= \sum_{i,j}p_{i}\RC_{ji}(f_{j}+\invBeta\ln p_{j})+ \sum_{i,j}p_{i}\RC_{ji}\gCi_{ji} % 
\nonumber\\
&= \invBeta \sum_{i,j}p_{i}\RC_{ji}[\ln p_{j}+\beta(f_{j}+ \gCi_{ji})] \label{eq:appzz11}.
\end{align}
The first line follows from definitions made in the main text. % 
Since $\RC$ is a rate matrix, $\sum_{i,j} p_i \RC_{ji} a_i =\sum_i p_i a_i \sum_{j} \RC_{ji}=0$ for any $\bm{a}$. This allows us to further rewrite \eqref{eq:appzz11} as 
\begin{align}
\FFdC{\pp}+\GGdC{\pp}&=\invBeta \sum_{i,j}p_{i}\RC_{ji}[\ln p_{j}-\ln p_i+\beta(-f_i + f_{j}+ \gCi_{ji})]\nonumber \\ % 
&= \invBeta \sum_{i,j}p_{i}\RC_{ji}[\ln p_{j}-\ln p_i-\ln (\RC_{ji}/\RC_{ij})] \label{eq:ldb-ineq-2}
\end{align}
In the second line, we plugged in the condition of LDB, % 
as expressed in~\eqref{eq:ldb} in the main text:
\begin{equation}
    \ln (\RC_{ji}/\RC_{ij})= \beta( f_{i}-f_{j} - \gCi_{ji})\qquad\text{for}\quad \RC_{ji}>0.\label{eq:appldb}
\end{equation} 
Using the definition of one-way fluxes due to control given a generic distribution $\pp$, $\JJcji = p_i \RC_{ji}$, and the Schnackenberg expression of EPR $\epr$, we rewrite~\eqref{eq:ldb-ineq-2} as 
\begin{align}
\FFdC{\pp}+\GGdC{\pp}= -\invBeta \sum_{i,j}p_{i}\RC_{ji}\ln \frac{p_j\RC_{ji}}{p_i\RC_{ij}}=-\invBeta\sum_{i\neq j} \JJcji\ln \frac{\JJcji}{\JJcij}=-\invBeta\epr(\JJc) \label{eq:ldb-ineq-3}\,.
\end{align}
In the main text, this equality is applied at the steady-state distribution $\pstBC$. However,  Eq.~\eqref{eq:ldb-ineq-3} is valid for all distributions $\pp$ as long as~\eqref{eq:appldb} holds.

To complete our derivation of \eqref{eq:appgg0}, % 
recall that $\RBC\pstBC=(\RB+\RC)\pstBC=0$ by definition of the steady-state $\pstBC$. At the same time, from the definition of $\FFdB{\pp}$ in~\eqref{eq:deltaF}, we have
\begin{align}
\FFdB{\pstBC} + \FFdC{\pstBC} = \sum_j (f_j+\invBeta\ln \pi_j) \sum_i (\RB_{ji}+\RC_{ji})\pi_i=0\qquad\implies\qquad \FFdB{\pstBC} = -\FFdC{\pstBC}.\label{eq:free-energies-at-steady-state}
\end{align}
Combining the definition of $\GGdBCbase$ with \eqref{eq:ldb-ineq-3} and \eqref{eq:free-energies-at-steady-state}   gives  Eq.~\eqref{eq:appgg0},
\begin{align}
\GGdBCbase= \GGdB{\pstBC} +\GGdC{\pstBC}=\GGdB{\pstBC}-\left[\FFdC{\pstBC}+\invBeta\epr(\JJc)\right] 
=\FFdB{\pstBC}+\GGdB{\pstBC}-\invBeta \epr(\JJc) \,.
\end{align}

Finally, note that the entropy production rate $\epr$ is always nonnegative: % 
\begin{align}
\epr(\JJc)=\frac{1}{2} \sum_{i \ne j}(p_{j}\RC_{ij}-p_{i}\RC_{ji})\ln \frac{p_j\RC_{ij}}{ p_{i}\RC_{ji}}\ge 0, \label{eq:ldb-ineq-99}
\end{align}
where the last inequality follows because the terms $p_{j}\RC_{ij}-p_{i}\RC_{ji}$ and $\ln \frac{p_j\RC_{ij}}{ p_{i}\RC_{ji}}$ have the same sign.  This can be seen as a statement of the Second Law.

Let us briefly comment on the physical meaning of our derivation. First, we note that LDB holds when the control dynamics exhibit a separation of timescales, with fast equilibration within each mesostate and slower transitions between mesostates~\cite{esposito2012stochastic}. 

Second, recall our assumption that the control processes do not interact directly with the external source of free energy, but only with the internal reservoir and heat bath. This assumption is formalized in the particular form of the LDB condition~\eqref{eq:appldb}.  It means that for the control processes, the combined ``system+internal reservoir'' may be treated as a single system
coupled only to a heat bath at inverse temperature $\beta$.  Hence, the rate of entropy production due to control is simply $\beta$ times the decrease of the combined free energy of the system and internal reservoir, as  in \eqref{eq:ldb-ineq-3}. Moreover, owing to Second Law, their combined
free energy cannot increase under control dynamics~\cite{horowitz2013imitating,brown_sivak_2020},
\begin{align}
\FFdC{\pp}+\GGdC{\pp}\le0 \qquad \text{for all}\quad \pp.\label{eq:secondlaw}
\end{align}

\subsection{Concavity of constrained optimization~\eqref{eq:Gmax00}}
\label{sec:optimization-properties-Gmax}

\global\long\def\Oobj{\mathscr{O}}% 
\global\long\def\YY{\boldsymbol{Y}}
\global\long\def\YYji{{Y}_{ji}}
\global\long\def\YYij{{Y}_{ij}}

Consider the objective $\FFdB{\pp}+\GGdB{\pp}-\epr(\JJ)/\beta$ in~\eqref{eq:Gmax00} in the main text. 
% : % 
% \begin{align}
% \FFdB{\pp}+\GGdB{\pp}-\epr(\JJ)/\beta.
% \end{align}
In Sec.~\ref{sec:optimization-properties-wmax} below, we show that the term $\FFdB{\pp}+\GGdB{\pp}$ is concave in $\pp$. Here we show that $\epr(\JJ)$ is convex in $\JJ$. This shows that the overall objective is concave in the pair $(\pp,\JJ)$.  
Consider any pair of feasible fluxes % 
$\JJ\neq \YY$ and $\lambda \in (0,1)$ and let $\JJ(\lambda)=\lambda\JJ+(1-\lambda)\YY$ indicate their convex mixture.
 To prove convexity, we write
\begin{align}
 \lambda \epr(\JJ)+(1-\lambda)\epr(\YY)&=
 \lambda \sum_{i \ne j}\JJji\ln \frac{\JJji}{\JJij}+(1-\lambda) \sum_{i \ne j} \YYji \ln \frac{\YYji}{\YYij}\nonumber \\
 & \ge \sum_{i \ne j}\JJji (\lambda)\ln \frac{\JJji (\lambda)}{\JJij(\lambda)} = \epr(\JJ(\lambda)).\label{eq:convexity-of-epr}
\end{align}
Here we used the log-sum inequality~\cite[Thm.~2.7.1]{cover_elements_2006},
\begin{align}
\lambda a\ln\frac{a}{a'}+(1-\lambda)b\ln\frac{b}{b'} \ge[\lambda a+(1-\lambda)b]\ln\frac{\lambda a+(1-\lambda)b}{\lambda a'+(1-\lambda)b'} \qquad\qquad\text{for all } a,a',b,b'\ge 0\,.
\label{eq:log-sum-ineq}
\end{align}

\subsection{Properties of the unconstrained optimization~\eqref{eq:wmax-1}}
\label{sec:optimization-properties-wmax}

\subsubsection{Concavity and existence of maximizer}

Consider the objective in \eqref{eq:wmax-1}, as equivalently written in \eqref{eq:wmax-lin-split}:
\begin{align}
    \label{eq:app4325}
\Lobj(\pp)=\FFdB{\pp}+\GGdB{\pp}=-\derivSRB{\pp}+\sum_{i}p_{i}\phi_{i},
\end{align}
where $\derivSRB{\pp}=-\sum_{i,j}p_{i}\RB_{ji}\ln p_{j}$
is the rate of increase of the Shannon entropy of distribution $\pp$ under the 
rate matrix $\RB$. Here we show that $\Lobj(\pp)$ is  concave, thus \eqref{eq:wmax-1} involves the maximization of a concave function. This 
is equivalent to the minimization of a convex function, making it convex optimization problem which can be solved
using standard numerical techniques. 
We also show the existence of a maximizer 
\[
\popt\in \operatorname*{arg\,max}_{\pp}\Lobj(\pp) \,.
\]

\emph{Existence of $\popt$}: The feasible set is
compact ($n$-dimensional probability simplex) and the objective is continuous, so the maximizer exists by the extreme value theorem. Also, the maximum value is finite (see \eqref{eq:upper-bound-deltaWmax-delta}).

\emph{Concavity}:  Consider any pair of distributions $\pp\ne\qq$
and $\lambda\in(0,1)$, and let $\pp(\lambda)=\lambda\pp+(1-\lambda)\qq $ indicate their convex mixture. We will show that the objective is concave:
\begin{equation}
\mathscr{L}(\pp(\lambda))\ge \lambda\mathscr{L}(\pp)+(1-\lambda)\mathscr{L}(\qq).\label{eq:app44423}
\end{equation}
Observe that $\sum_i p_i \phi_i$ is linear in $\pp$, thus we can use \eqref{eq:app4325} to rearrange  \eqref{eq:app44423} as
\begin{equation}
\derivSRB{\pp(\lambda)}\le \lambda\derivSRB{\pp}+(1-\lambda)\derivSRB{\qq}.\label{eq:app44423v4}
\end{equation}
To prove \eqref{eq:app44423v4}, we first rewrite the rate of decrease of the Shannon entropy as
\begin{equation}
\derivSRB{\pp}=-\sum_{i,j}p_{i}\RB_{ji}\ln p_{j}=-\sum_{i,j}p_{i}\RB_{ji}\ln\frac{p_{j}}{p_{i}}=-\sum_{i\ne j}p_{i}\RB_{ji}\ln\frac{p_{j}}{p_{i}}=\sum_{i\ne j}p_{i}\RB_{ji}\ln\frac{p_{i}}{p_{j}}.\label{eq:app4444}
\end{equation}
Here we first used that $\sum_{i,j} p_i \RB_{ji} a_i=\sum_i a_i p_i \sum_j \RB_{ji}=0$ for any $\bm{a}$, then used that the diagonal terms of the sum ($i=j$) vanish. Next, we use \eqref{eq:app4444} to prove \eqref{eq:app44423v4} as
\begin{align}
\derivSRB{\pp(\lambda)}=\sum_{i\ne j}p_{i}(\lambda)\RB_{ji}\ln\frac{p_{i}(\lambda)}{p_{j}(\lambda)}&=\sum_{i\ne j}p_{i}(\lambda)\RB_{ji}\ln\frac{p_{i}(\lambda)\RB_{ji}}{p_{j}(\lambda)\RB_{ji}}\nonumber \\
&\le \sum_{i\ne j}\left[\lambda p_{i}\RB_{ji}\ln\frac{p_{i}\RB_{ji}}{p_{j}\RB_{ji}}+(1-\lambda)q_{i}\RB_{ji}\ln\frac{q_{i}\RB_{ji}}{q_{j}\RB_{ji}}\right] \label{eq:bdg33}\\
&= \lambda\sum_{i\ne j}p_{i}\RB_{ji}\ln\frac{p_{j}}{p_{i}}+(1-\lambda)\sum_{i\ne j}q_{i}\RB_{ji}\ln\frac{q_{j}}{q_{i}}  \nonumber\\
&=\lambda\derivSRB{\pp}+(1-\lambda)\derivSRB{\qq}
\,. \nonumber 
\end{align}
The second line follows  from the log-sum inequality~\eqref{eq:log-sum-ineq}.

\subsubsection{Strict concavity and uniqueness of maximizer under irreducibility assumption}

We can prove further results under the assumption that the baseline rate matrix $\RB$ is irreducible (i.e., there is a path of non-zero transitions connecting any two states) and has a steady-state with full support ($\pstBCnb_i>0$ for all $i$). 
In graph theoretic
terms, this means that the graph of allowed transitions under $\RB$
is strongly connected.  Specifically, we prove that our objective is strictly concave, and that the optimizer $\popt$ is unique and has full support.

\emph{Strict concavity}:
The log-sum inequality \eqref{eq:log-sum-ineq} is strict when $a/a'\ne b/b'$~\cite[Thm.~2.7.1]{cover_elements_2006}. In our case, \eqref{eq:bdg33} 
is strict for any pair of states $j\ne i$
that have $\RB_{ji}>0$ and
\begin{align}
    \label{eq:app353}
\frac{p_{i}\RB_{ji}}{p_{j}\RB_{ji}}\ne\frac{q_{i}\RB_{ji}}{q_{j}\RB_{ji}}.
\end{align}
Such a pair of states must exist, as we now prove by contradiction. 
Suppose to the contrary that \eqref{eq:app353} is an equality 
for all $i\ne j$ where $\RB_{ji}>0$. Consider a walk on the graph
defined by the allowed transitions in $\RB$ --- that is, a sequence
of  states $i_{0},i_{1},i_{2},\dots$ such that $\RB_{i_{k+1}i_{k}}>0$. We would then have that $p_{i_{0}}/p_{i_{1}}=q_{i_{0}}/q_{i_{1}},p_{i_{1}}/p_{i_{2}}=q_{i_{1}}/q_{i_{2}},\dots$,
hence also (by multiplication) that $p_{i_{0}}/p_{i_{2}}=q_{i_{0}}/q_{i_{2}},\dots$.
Now, since $\RB$ is irreducible, any state can be reached from any
other, implying that $p_{i}/p_{j}=q_{i}/q_{j}$ for all $i\ne j$, hence $\pp \propto \qq$. 
Since probabilities are nonnegative and normalized to sum to 1, this
can only hold when $\pp=\qq$, contradicting our assumption that $\pp\ne\qq$ in  \eqref{eq:bdg33}. 

To summarize, there must be some $i \ne j$ such that the corresponding term in \eqref{eq:bdg33}  is a strict inequality, therefore
\begin{equation}
\mathscr{L}(\pp(\lambda))>\lambda\mathscr{L}(\pp)+(1-\lambda)\mathscr{L}(\qq).\label{eq:app44423v2}
\end{equation}

\emph{Uniqueness of maximizer}: Suppose that
 there are two different maximizers, $\mathscr{L}(\popt_{(1)})=\mathscr{L}(\popt_{(2)})=\Gmax$.
Then, \eqref{eq:app44423v2}  implies that the convex mixture $\lambda\popt_{(1)}+(1-\lambda)\popt_{(2)}$ would achieve a larger value than $\Gmax$, leading to a contradiction.

\emph{Full support of $\popt$}:
Suppose that the maximizer $\popt$ did not have full support, so
that some states have 0 probability. Then, there must be some pair  $i\ne j$ such that $p_{i}^{*}>0$ and $p_{j}^{*}=0$ and
$\RB_{ji}>0$ (otherwise the graph of allowed transitions would not be
strongly connected). But for this pair, $p_{i}^{*}\RB_{ji}\ln({p_{j}^{*}}/{p_{i}^{*}})=-\infty$. Plugging into \eqref{eq:app4444} implies that $\Lobj(\popt)=-\infty$ and contradicting the fact that
$\popt$ is a maximizer. Hence, $\popt$ must have full support.

\subsection{Derivation and achievability of unconstrained optimization~\eqref{eq:wmax-1}}
\label{sec:achievability}

In this section, we derive the unconstrained optimization~\eqref{eq:wmax-1} from the constrained optimization~\eqref{eq:Gmax00}.  

First, observe that  that
the maximum in \eqref{eq:Gmax00} is always less than the maximum in \eqref{eq:wmax-1}:
\begin{gather*}
\Gmax:=  \sup_{(\pp, \JJ) \in \feasibleset  : \incmatrix \JJ =-\RB \pp } \, \FFdB{\pp} + \GGdB{\pp} - \epr(\JJ)/\beta  \quad \le \quad  
\Gmax^\prime :=\max_{\pp}\,\FFdB{\pp}+\GGdB{\pp}
\,.
\end{gather*}
This follows from nonnegativity of $\epr(\JJ)$ and the fact that \eqref{eq:wmax-1} has less constraints than \eqref{eq:Gmax00}. 

We will show that $\Gmax=\Gmax^\prime$ as long as the feasible set $\feasibleset$ is not too constrained. In addition to the basic normalization and nonnegativity constraints on $\pp$ and $\JJ$, we allow $\feasibleset$ to encode a set of topological constraints like
\begin{align}
    \JJji = 0 \qquad\text{if }G_{ji}=0\,,
    \label{eq:appGmxdef}
\end{align}
for $i \ne j$, where $G$ is a symmetric matrix that determines the topology of the control dynamics ($G_{ji}=G_{ij}=1$ when the transition $i\leftrightarrow j$ is allowed, and otherwise $G_{ji}=G_{ij}=0$). We assume that $G$ is connected and contains all $n$ states. Of course, we may have $G_{ji}=1$ for all $i\ne j$, in which case no topological constraints are imposed. We assume that $\Lambda$ includes no other constraints.

Given these assumptions, we construct a sequence of control process $(\RC(\kappa),\gC)$ which satisfy LDB such that: 
\begin{enumerate}
    \item The combined steady state $\pstBC(\kappa)$ of $\RB+\RC(\kappa)$ and  control fluxes $\JJcij(\kappa) = \pstBCnb_j(\kappa) \RC_{ij}(\kappa)$ belong to $\feasibleset$ for all $\kappa$.
    \item The combined steady state  obeys $\lim_{\kappa \to \infty}\pstBC(\kappa) = \popt$, where $\popt$ is the optimizer of \eqref{eq:wmax-1}.
    \item The steady-state EPR  vanishes $\lim_{\kappa \to 0} \epr(\kappa)= 0$, therefore the  steady-state harvesting rate  obeys $$\Gmax \ge \lim_{\kappa \to \infty} \GGdBCbase(\kappa)=\Gmax^\prime.$$
\end{enumerate}

Our proof technique is related to the idea sketched out in Ref.~\citep{horowitz_minimum_2017}, which studied the minimal cost of maintaining a nonequilibrium steady state. It is also related to constructions from the literature on counterdiabatic driving \cite{ilker2022shortcuts}.

\subsubsection{Construction of optimal control}

\newcommand\baseRC{B}
\newcommand{\RCinv}{\baseRC^+}

We define a sequence of control processes parameterized by $\kappa >0$. 
For each control process, the free energy exchanges with the internal reservoir are set to 
\begin{align}
\gCi_{ji}=f_{i}-f_{j}+\invBeta (\ln p_{i}^{*}-\ln p_{j}^{*}).
\end{align}
The control rate matrix $\RC(\kappa)$ is defined by scaling a given rate matrix $\baseRC$
\begin{align}
    \RC_{ji}(\kappa) := \rateparam \baseRC_{ji} \qquad \qquad \baseRC_{ji} := \frac{G_{ji}}{1+e^{\beta( f_{i}-f_{j} - \gCi_{ji})}}=\rateparam \frac{G_{ji}}{1+p_j^*/ p_i^*}\,.\label{eq:rcparameterization}
\end{align}
Here $\rateparam \ge 0$ is an overall rate constant and $G$ is the matrix discussed in \eqref{eq:appGmxdef}.  
The particular choice of the topology encoded in $G$ is arbitrary, given our assumption that it is connected and contains all $n$ states which guarantees that  $\RC$ is irreducible. % 
It can be verified that the rate matrix $\RC(\kappa)$ defined in \eqref{eq:rcparameterization} obeys LDB~\eqref{eq:ldb}. It also satisfies 
detailed
balance (DB) for the distribution $\popt$,
\begin{align}
\RC_{ji}(\rateparam)\popti_{i}=\RC_{ij}(\rateparam)\popti_{j}, \label{eq:achievability-DB-control}
\end{align}
This implies that $\popt$ is the unique steady-state distribution, $\RC(\rateparam)\popt=0$.

We note that there are various other types of rate matrices that can be used in the construction, as long as they satisfy LDB and DB. However, the parameterization used in \eqref{eq:rcparameterization} is common in the literature~\cite{ferreira2006model}.

\subsubsection{Achievability of the steady state $\popt$}

We show that $\pstBC(\kappa)$, the steady-state of the combined rate matrix of $\RBC({\rateparam}):=\RB +\RC(\rateparam)$, approaches $\popt$ in the limit of fast control (large $\rateparam$). 
Since $[\RB +\RC(\rateparam)]\pstBCa=0$ and $\RC(\kappa)\popt=0$, we have % 
\begin{align}
\frac{1}{\rateparam}\RB\pstBCa=-\frac{1}{\kappa}\RC(\kappa)\pstBCa=\frac{1}{\kappa}\RC(\kappa)\left(\popt-\pstBCa\right)=\baseRC(\popt-\pstBCa). 
\end{align}
Rearranging and applying the Moore-Penrose inverse $\RCinv$ to both sides gives
\begin{align}
\frac{1}{\rateparam}\RCinv \RB \pstBCa =\RCinv\baseRC\left(\popt-\pstBCa\right)\,.
\label{eq:reach-proof-matrix-eq1}
\end{align}
Since $\baseRC$ is irreducible by construction, $\RCinv\baseRC=I-\popt\boldsymbol{1}^T$. Plugging into \eqref{eq:reach-proof-matrix-eq1}, we obtain
\begin{align}
\frac{1}{\rateparam}\RCinv \RB \pstBCa=(I-\popt\boldsymbol{1}^T)(\popt-\pstBCa)=\popt-\pstBCa \label{eq:vdsff3},
\end{align}
which follows from $\boldsymbol{1}^T\popt=\boldsymbol{1}^T \pstBCa=1$. Taking norms gives
\begin{align}
    \left\Vert \popt-\pstBCa \right\Vert=\frac{1}{\rateparam}\Vert \RCinv \RB \pstBCa \Vert\leq \frac{1}{\rateparam}\Vert \RCinv \RB\Vert  \Vert \pstBCa \Vert\leq \frac{1}{\rateparam}\Vert \RCinv \RB\Vert , \label{eq:achievability-bounded-norm}
\end{align}
where in the last step we used that the norm of any probability distribution is bounded by 1. This shows that $\lim_{\kappa \to \infty}\left\Vert \popt-\pstBCa \right\Vert = 0$, meaning that the combined steady-state converges to $\popt$.

\subsubsection{EPR vanishes}

\newcommand\pstBCai{\pstBCnb_i(\rateparam)}
\newcommand\pstBCaj{\pstBCnb_j(\rateparam)}

The steady-state EPR incurred by rate matrix $\RC(\kappa)=\kappa \baseRC$ is 
\begin{align}
\epr(\kappa) = \frac{\rateparam}{2 } \sum_{i \ne j}(\pstBCaj\baseRC_{ij}-\pstBCai\baseRC_{ji})\ln \frac{\pstBCaj\baseRC_{ij}}{ \pstBCai\baseRC_{ji}}\,.\label{eq:bdfd3}
\end{align}
Next, we use \eqref{eq:vdsff3} to write $\popt=\pstBCa +\frac{1}{\rateparam}\RCinv\RB\pstBCa$. 
Using the DB condition for $\popt$ with respect to $\RC(\kappa)=\kappa\baseRC$ obtained in \eqref{eq:achievability-DB-control}, we rearrange terms to find:
\begin{align}
\pstBCai \baseRC_{ji}-\pstBCaj \baseRC_{ij}=\frac{1}{\rateparam}\left[\baseRC_{ij}\left(\RCinv\RB\pstBCa\right)_j-\baseRC_{ji}\left(\RCinv\RB\pstBCa\right)_i\right] \label{eq:fluxes-with-pstBCa}
\end{align}
Plugging into \eqref{eq:bdfd3} gives
\begin{align}
\epr(\kappa)=\frac{1}{2}\sum_{i\neq j}\left[\baseRC_{ji}\left(\RCinv\RB\pstBCa\right)_i-\baseRC_{ij}\left(\RCinv\RB\pstBCa\right)_j\right]\ln \frac{\pstBCaj \baseRC_{ij}}{ \pstBCai \baseRC_{ji}}.\label{eq:bfdg2}
\end{align}
Taking the limit $\rateparam\to \infty$ shows that EPR vanishes:
\begin{align*}
\lim_{\rateparam\to \infty} \epr(\kappa)= \frac{1}{2}\sum_{i\neq j}\left[\baseRC_{ji}\left(\RCinv\RB\popt\right)_i-\baseRC_{ij}\left(\RCinv\RB\popt\right)_j\right]\ln \frac{p^*_j \RC_{ij}}{ p^*_{i} \RC_{ji}}=0,
\end{align*}
where we used $\lim_{\rateparam\to\infty}\pstBCa=\popt$ and the DB condition~\eqref{eq:achievability-DB-control} in the logarithmic factor.

\clearpage
\section{Bacteriorhodopsin model}
\label{sec:br}

\subsection{Details of model}
\label{sec:brmodeldetails}
\newstuff{

Here we provide thermodynamic and kinetic details of the bacteriorhodopsin model analyzed in the main text.

The thermodynamic parameters are taken from Ref.~\cite{varo1991thermodynamics}, which reports \emph{in vitro} measurements of internal free energies at $293^\circ\,\text{K}=20^\circ \,\text{C}$. Based on Fig.~7 in that paper, we use the following internal free energies for the 6 cycle states: % 
\begin{align}
\bm{f} &\equiv && (f_K,&& f_L, &&f_{M_1}, &&f_{M_2}, &&f_N, &&f_O &&\!\!\!\!\!\!) \nonumber\\
&= && (34.41, && 27.96, && 29.57, && 13.98, && 13.17, && 14.78&&\!\!\!\!\!\!)\;\;\;\text{kJ mol}^{-1}\nonumber \\
&=&&(5.71\times 10^{-20}, &&\!\!\!4.64\times 10^{-20}, &&\!\!\!4.91\times 10^{-20}, && \!\!\!2.32\times 10^{-20},&&\!\!\! 2.19\times 10^{-20},&& \!\!\!2.45\times 10^{-20}&&\!\!\!\!\!\!) \;\;\;\text{joules}
\label{eq:appfdef}
\end{align}
All values are referenced from the ground state, in other words the zero point refers to the internal free energy of the ground state $\mathrm{bR}$~\cite{varo1991thermodynamics}. However, since we coarse-grain the transitions $O\to \mathrm{bR}$ and $\mathrm{bR}\to K$ into a single transition $O\to K$, we do not include this ground state in our model.  See Fig.~\ref{fig:br-cg} for an illustration of the coarse-graining.

\begin{figure}[h]
\includegraphics[width=.6\columnwidth]{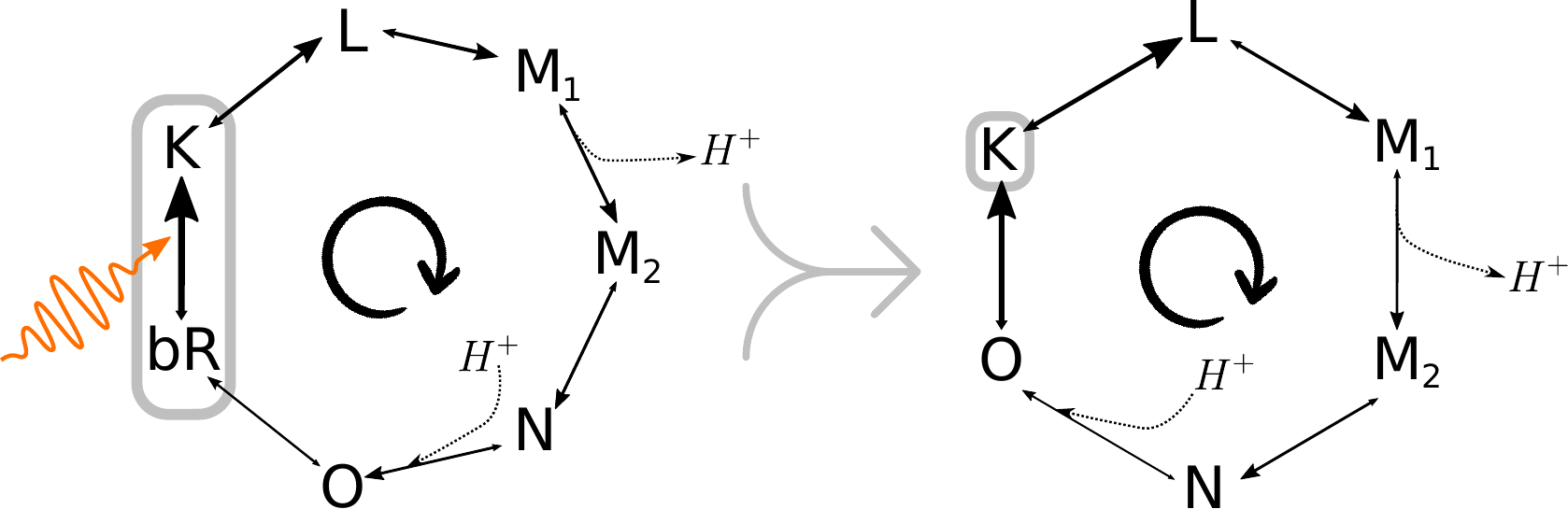}
\caption{\label{fig:br-cg}
\figurefontcmd
\newstuff{Bacteriorhodopsin cycle. Using the irreversible character and fast speed of the $\mathrm{bR}\to K$ transition, we coarse-grain the cycle from the original seven-state system (left) to a six-state one (right). Circled region in gray shows the states and transition that is coarse-grained into a single state. 
}}
\end{figure}

We use $g_{ji}$ to indicate the free energy transferred to the internal reservoir by the jump from state $i\in \{K, L, M_1, M_2, N, O\}$ to state $j\in \{K, L, M_1, M_2, N, O\}$.  Only transitions between neighboring states in the cycle are permitted: $K\leftrightarrow L$, $L\leftrightarrow M_1$, $M_1\leftrightarrow M_2$, $M_2\leftrightarrow N$, $N\leftrightarrow O$, and $O\leftrightarrow K$. The values of $g_{ji}$ are zero for all transitions except for the transition $M_1 \leftrightarrow M_2$, corresponding to proton transport across the membrane. This value is determined via Eq.~\eqref{eq:pstep} as a function of the electrochemical potential (a.k.a. proton motive force), which includes contributions from electrical potential and pH difference across the membrane. To summarize:
\begin{align}
\begin{gathered}
g_{LK} = g_{KL} = g_{M_1 L} = g_{L M_1} = g_{N M_2} = g_{M_2 N } = g_{O N} = g_{N O} = g_{K O} = g_{O K } = 0 \\
g_{M_2  M_1} = -g_{M_1 M_2} =\Delta \mathrm{p} := e\Delta \psi-(\ln 10)\beta^{-1}\Delta\text{pH} \qquad\text{joules}   
\end{gathered}
\label{eq:appgvals}
\end{align}
We emphasize that $\Delta\psi$ is the  membrane electrical potential in volts, where $\Delta\psi\ge0$ indicates that the inside
is more negative. $\Delta\text{pH}$ is the membrane pH difference, where $\Delta\text{pH}\le0$ indicates that the inside is more basic.  
No free energy is exchanged in the processes internal to each cycle state, so 
\begin{align}
\ggB=0 \,.
\end{align}

The dynamics are represented by the system's rate matrix. We parameterize the  transition rate for the jump from state $i\in \{K, L, M_1, M_2, N, O\}$ to neighboring state $j\in \{K, L, M_1, M_2, N, O\}$  as
\begin{align}
R_{ji} = \frac{\kappa_{ji}}{1 + e^{-\Delta s^\mathrm{tot}_{ji}}}\,,\label{eq:br-appldb}
\end{align}
where $\kappa_{ji}=\kappa_{ij}$ is the relaxation rate for the transition $i\leftrightarrow j$, and $\Delta s^\mathrm{tot}_{ji}=-\Delta s^\mathrm{tot}_{ij}$ is the entropy produced during the transition $i \to j$. 
This is a common parametrization which guarantees that the rates satisfy local detailed balance (LDB)~\cite{ferreira2006model}. 
The entropy production for each jump $i\to j$ is
\begin{align}
\Delta s^\mathrm{tot}_{ji} = \beta(f_i -f_j - g_{ji} + m_{ji}) \,.
\end{align}
where $f_i$ refers to internal free energies in joules~\eqref{eq:appfdef}, $\beta = 1/k_B T$ and $k_B= 1.38\times10^{-23}\ \text{joules}\cdot\left(^{\circ} \text{ K}\right)^{-1}$ is Boltzmann's constant.  
As above $g_{ji}=-g_{ij}$ is the increase of free energy of the internal reservoir, while $m_{ji}=-m_{ij}$ is the decrease of free energy in the external source during the transition $i \to j$. In fact, $m_{ji}=0$ for all $i,j$ except the transition $O\leftrightarrow K$, which corresponds to photon absorption. The energy absorbed from the photon is $hc/\lambda$ joules, where $h$ is Planck's constant, $c$ is speed of light, and $\lambda$ is the photon wavelength.  
We use a physiologically plausible value of 580 nm~\cite{lanyi1986halorhodopsin}. 
At this wavelength and temperature of $293^\circ$ K,  
\begin{align}
m_{K O} = -m_{OK} = hc/\lambda \;\;\text{joules} \quad \approx \quad 84 \,k_B T \,.
\end{align}
Observe that the transition $O\to K$ is highly irreversible ($r_{KO} \gg r_{OK}$). In fact, our results are essentially the same regardless of whether the transition rate for this step is computed using \eqref{eq:br-appldb} or made absolutely irreversible. 

The relaxation kinetics that enter into \eqref{eq:br-appldb} are taken from Table 1 in Ref.~\cite{lorenz-fonfriaSpectroscopicKineticEvidence2009}. We use the geometric mean of the upper and lower estimates of $k_\mathrm{relax}^{-1}$ to get the following relaxation rates (in sec$^{-1}$):
\begin{align}
\begin{aligned}
&\kappa_{LK}      && = \kappa_{KL} &= 2.57\times 10^5 
&\qquad\qquad\kappa_{N M_2} &&= \kappa_{M_2 N} &=  7.22\times 10^2 
\\
&\kappa_{M_1 L}      && = \kappa_{L M_1} &= 2.55\times 10^4 
&\qquad\qquad\kappa_{O N }  &&= \kappa_{N O} &=5.10\times 10^2 \\
&\kappa_{M_2 M_1}   && = \kappa_{M_1 M_2} &= 5.42\times 10^3 
&\qquad\qquad\kappa_{KO}      &&= \kappa_{O K} &=  1.28\times 10^2 
\end{aligned} 
\end{align}

\vspace{10pt}
\newstuff{
For concreteness, here we provide the numerical values of entropy production for each jump $i\to j$, at \emph{in vivo} potential $\Delta \psi = 120$ mV:
\begin{align}
 & && (\Delta s^\mathrm{tot}_{LK},&& \Delta s^\mathrm{tot}_{M_1 L}, &&\Delta s^\mathrm{tot}_{M_2M_1}, &&\Delta s^\mathrm{tot}_{N{M_2}}, &&\Delta s^\mathrm{tot}_{ON}, &&\Delta s^\mathrm{tot}_{KO} &&\!\!\!\!\!\!) \nonumber\\
&= && (2.65, && -0.66, && 0.21, &&0.33, &&-0.66, && 76.62 &&\!\!\!\!\!\!)
\label{eq:numfeff}
\end{align}
The resulting rate matrix, again at the \emph{in vivo} potential, is (in sec$^{-1}$)
\newcommand\mycolwidth{5.5em}
\begin{align}
R = \,& 
\left[
\begin{array}{rrrrrr}
-2.40 \times 10^{5} & 1.70 \times 10^{4} & 0 & 0 & 0 & 1.28 \times 10^{2} \\
2.40 \times 10^{5} & -2.57 \times 10^{4} & 1.68 \times 10^{4} & 0 & 0 & 0 \\
0 & 8.69 \times 10^{3} & -1.99 \times 10^{4} & 2.35 \times 10^{3} & 0 & 0 \\
0 & 0 & 3.07 \times 10^{3} & -2.77 \times 10^{3} & 3.01 \times 10^{2} & 0 \\
0 & 0 & 0 & 4.20 \times 10^{2} & -4.75 \times 10^{2} & 3.37 \times 10^{2} \\
6.78 \times 10^{-32} & 0 & 0 & 0 & 1.74 \times 10^{2} & -4.65 \times 10^{2}
\end{array}
\right]\,.\label{eq:numratematrix}
\\
&\hspace{1em}\underbrace{\hspace{\mycolwidth}}_{K}
\underbrace{\hspace{\mycolwidth}}_{L}
\underbrace{\hspace{\mycolwidth}}_{M_1}
\underbrace{\hspace{\mycolwidth}}_{M_2}
\underbrace{\hspace{\mycolwidth}}_{N}
\underbrace{\hspace{\mycolwidth}}_{O}
\nonumber
\end{align}
}

\subsection{Details of numerical analysis in Fig.~\ref{fig:br-data}}

To calculate the curves plotted in Fig.~\ref{fig:br-data}, we first define the ordered set of transitions 
\begin{align}\mathcal{T}&:=% 
(K\leftrightarrow L, L\leftrightarrow M_1, M_1\leftrightarrow M_2, M_2 \leftrightarrow N, N\leftrightarrow O)\nonumber\end{align}
Then, for each value of the electrical membrane potential $\Delta \psi$ we perform the following:
\begin{enumerate}
\item Use the transition rates from \eqref{eq:br-appldb} to define the `total' (baseline-and-control) rate matrix $R$ and numerically solve for $R\pstBC = 0$ 
to obtain $\pstBC$, which is unique since $R$ is irreducible (all states are connected, see Fig. \ref{fig:br-cg}).
\item Use the transition rates from \eqref{eq:br-appldb} and values of $\bm{g}$ from \eqref{eq:appgvals}  to calculate the total harvesting rate as
\begin{align}
\GGdBCbase=\sum_{i,j}\pi_{i}\RBC_{ji} g_{ji} \,.
\end{align}
The total harvesting rate is plotted as the thick black line in Fig.~\ref{fig:br-data}.
\item For each transition $t\in\mathcal{T}$ that acts as control: % 
\begin{enumerate} 
\item
Define the baseline transition matrix $\RB$ by removing the chosen transition from $R$. As standard, the diagonal entries $\RB_{ii}$ are determined by $\RB_{ii}=-\sum_{j} \RB_{ji}$.
\item Define optimization parameters $(\pp, \JJ)$ with $\pp\in\mathbb{R}^n$ and $\JJ$ a vector in $\mathbb{R}^{n^2}$ constrained such that $\sum_i p_i = 1$, $p_i\geq 0$ $\forall i$, $\mathbb{B}\JJ = -R^b\pp$ where $\incmatrix\in \mathbb{R}^{n\times n^2}$ 
is the incidence matrix with entries $\incmatrix_{k,ij}=\delta_{ki}-\delta_{kj}$. Finally, we make all $\JJ_{ji}=\JJ_{ij}=0$ for $\{i,j\}\neq t$. 
\item Use the transition rates and free energy values of $\bm{f}$ from \eqref{eq:appfdef} to fix the ``baseline harvesting rate'' and the ``increase of system free energy'' (Eq.~\eqref{eq:deltaF}) functions:
\begin{align}
\GGdB{\pp} &=\sum_{i,j}p_{i}\RB_{ji} g_{ji}\\
\FFdB{\pp} &=\sum_{i,j}p_{i}\RB_{ji}(f_{j}+\invBeta\ln p_{j}).
\end{align}
\item Plug $\FFdB{\pp}$, $\GGdB{\pp}$ and $\epr(\JJ)$ into Eq.~\eqref{eq:Gmax00} to solve for $\popt$, $\JJ^*$ and $\Gmax$ (color-coded line in Fig.~\ref{fig:br-data}) using numerical optimization.
\end{enumerate}
\end{enumerate}

We emphasize that in the data shown in Fig.~\ref{fig:br-data} we add no extra constraints. However, in the remaining of this supplementary material, we consider various additional linear constraints pertaining to activity, thermodynamic affinity and kinetic limitations. Such constraints are imposed at Step 3.b and determine the feasibility set $\Lambda$ in Eq.~\eqref{eq:Gmax00}. % 
}

We make three final observations regarding our bacteriorhodopsin model. First, our analysis assumes that control processes can be added without affecting membrane parameters, such as $\Delta\mathrm{pH}$ and $\Delta\psi$. In practice, this may be justified by homeostatic mechanisms.  For instance, excess proton pumping produced by the introduction of control may be balanced by up-regulation of ATPase, which consumes the protons while making ATP. 

Second, in order to consistent with LDB~\eqref{eq:ldb}, changing the  rate of control transition  $i\to j$, while keeping the internal free energy values fixed, may require changing the free energy used by that control transition $\gCi_{ji}$. For transitions coupled to the membrane potential (such as $M_1\leftrightarrow M_2$), this can be achieved by manipulating $\Delta \psi$ or $\Delta \text{pH}$, which here act as external parameters (see~\eqref{eq:pstep}). However, for a transition that is not coupled to the membrane potential, manipulating $\gCi_{ji}$ may require, for example, the consumption of a chemical fuel such as ATP. The strength of driving can be modulated by regulating the nonequilibrium concentration of the chemical fuel.

Third, our model reproduces steady-state currents and parameter dependence (such as membrane potential) that agree with reported data, at least to a first approximation. At the same time, it is worth noting that the equilibrium constants reported in our source of kinetic data~\cite{lorenz-fonfriaSpectroscopicKineticEvidence2009} differ somewhat from the free energy changes reported in our source of thermodynamic data~\cite{varo1991thermodynamics}. This may be due to different experimental or estimation methods, or because some reaction steps are not  approximated well as elementary transitions.

\newstuff{
\subsection{Additional analysis: reprotonation step as control}
\label{app:reprotonationascontrol}

In the analysis in the main text, we observe that near the {\em in vivo} membrane potential, 
the reprotonation reaction ($N\leftrightarrow O$) is the least efficient of all five transitions in $\mathcal{T}$. In this section, we study this transition as control with additional constraints. 

Our analysis can be motivated in the context of synthetic or natural selection for increased output in bacteriorhodopsin~\cite{miller1990kinetic,seitz2000kinetic,wise2002optimization,hillebrecht2004directed}.  Suppose that the genotype-phenotype map of bacteriorhodopsin is sufficiently modular such that the thermodynamic and kinetic properties of individual transitions in the cycle can be separately manipulated by mutations. In fact, in case of the reprotonation transition $N\leftrightarrow O$,  there is evidence that a single amino acid in the bacteriorhodopsin protein specifically and effectively changes the kinetics of that transition~\cite{miller1990kinetic}. In this setting, we ask to what extend the existing $N\leftrightarrow O$ transition is optimal. At the same time, we illustrate how underlying thermodynamic and kinetic constraints, e.g., as might arise from underlying diffusion timescales and biochemistry, can be incorporated when quantifying optimality.

In concrete terms, we let $\RB_{NO}=\RB_{ON}=0$ and fix the flux parameters of the optimization problem \eqref{eq:Gmax00} as $\JJij=\JJji=0$ for all $\{i,j\}\neq\{N,O\}$. Formally, we consider the following feasibility set in \eqref{eq:Gmax00}:
\begin{align} % 
\Lambda^{NO}:=\{(\pp,\JJ) : \pp \ge 0, \bm{1}^T\pp = 1, \JJ \ge 0, J_{ij}=0 \text{ if } (i,j)\not \in \{{(N,O)},{(O,N)}\}\}_{} \,.
\end{align}
along with other constraints discussed below. 

To compare the result of optimization to the actual system, we also consider the actual transition as the control, defined via:
\begin{align}
\RC_{O N} = \frac{\kappa_{O N}}{1 + e^{-\Delta s^\mathrm{tot}_{O N}}} 
\qquad\qquad \RC_{N O} = \frac{\kappa_{N O}}{1 + e^{-\Delta s^\mathrm{tot}_{O N}}} 
\end{align}
with all other $\RC_{ij}=0$. The baseline-and-control dynamics correspond to the actual bacteriorhodopsin system described in Sec.~\ref{sec:brmodeldetails}.  See also Fig.~\ref{fig:br-sm-baseline-control} (left).

For Sections~\ref{app:reprotonationascontrol} and~\ref{app:protonpumpascontrol} in this SM, it is useful to define the  \emph{current} (i.e., net flux) as
\begin{align}
\mathcal{J}=J_{NO}^*-J_{ON}^*
\label{eq:currentdefinition}
\end{align} 
where $J_{NO}^*$ and $J_{ON}^*$ are the optimal fluxes found by our optimization problem~\eqref{eq:Gmax00}. In the setting of these sections, where the $N\leftrightarrow O$ control transition is the missing step of the unicyclic bacteriorhodopsin system, $\mathcal{J}$ is the cyclic current in the presence of the baseline and the optimal control transition.

\begin{figure}
\includegraphics[width=\columnwidth]{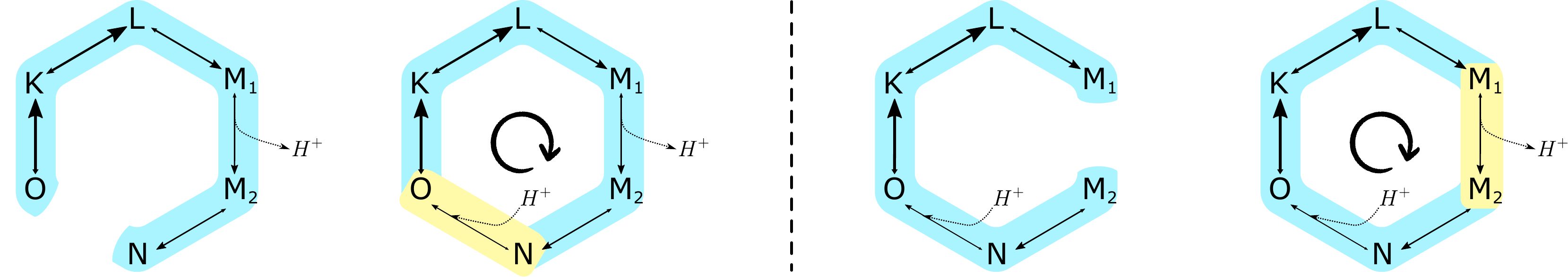}
\caption{\label{fig:br-sm-baseline-control}
\figurefontcmd
\newstuff{Depiction of choices for baseline and control for Sections~\ref{app:reprotonationascontrol} (left) and~\ref{app:protonpumpascontrol} (right) of the SM. In each case, the baseline (shaded blue) does not include the transition of reprotonation ($N\leftrightarrow O$) and proton-pumping ($M_1\leftrightarrow M_2$), respectively; instead, these are treated as control (shaded yellow). For both scenarios, as in the analysis done in the main text, note that no cyclic current circulates under baseline dynamics but does so under baseline-and-control (i.e. the actual) dynamics.
}
}
\end{figure}

\subsubsection{Constrained activity and affinity}

We optimize the maximum harvesting rate with respect to the transition $N\leftrightarrow O$, while imposing an additional constraint on the ``dynamical activity'' of the fluxes $N\to O$ and $O\to N$. This defines a smaller feasibility set in \eqref{eq:Gmax00}:
\begin{align}\label{eq:activity-constrain0}
\Lambda_a^{NO}:=\Lambda^{NO} \cap \{ (\pp,\JJ) : J_{NO}+J_{ON}\leq a\},
\end{align}
for some value of $a$ in units of $\text{sec}^{-1}$. 
This upper-bounds the dynamical activity of our control transition such that fluxes cannot be arbitrarily large. 

\begin{figure}
\includegraphics[width=.7\columnwidth]{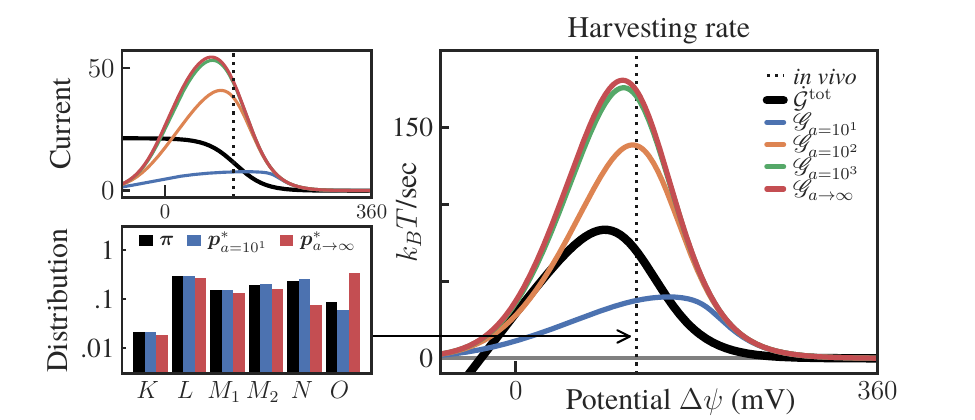}
\caption{\label{fig:N-O-activity-constraint}
\figurefontcmd
\newstuff{Right: actual  harvesting rate (in black) and  optimized harvesting rate for controlled transition $N\leftrightarrow O$ at different values of activity constraint $a$. Dashed vertical line shows \emph{in vivo} membrane potential ($120\text{mV}$). Left panel shows the respective currents (upper plot) and the distribution at \emph{in vivo} value of $\Delta \psi$ (lower plot) for the actual system versus the optimized solutions under $a=10^1$ and $a\to \infty$, in units of $\text{sec}^{-1}$.
}
}
\end{figure}

The result of this analysis are shown in Fig.~\eqref{fig:N-O-activity-constraint}. 
We begin by noting that, in some instances in the right plot in this figure, the optimized constrained control yields a lower harvesting rate than that of the actual system (black line). For example, if $a= 10^1$ at around \emph{in vivo} values, the actual harvesting rate is roughly double the one achieved by the constrained $\Gmax_{a= 10^1}$ (blue line). This is not too surprising since, from our computation of the actual harvesting rate at \emph{in vivo} values, we read off (in $\text{sec}^{-1}$):
$$J^\text{\emph{in\,vivo}}_{ON}=\pi_{N}R_{ON}\simeq 41.6\qquad J^\text{\emph{in\,vivo}}_{NO}=\pi_{O}R_{NO}\simeq 30.1 $$
Thus, constraining the activity as in~\eqref{eq:activity-constrain0} limits the ability of the optimized system to amplify the current above the actual \emph{in vivo} current (see top left plot in Fig.~\ref{fig:N-O-activity-constraint}). In contrast, increasing $a$ approaches the optimal solution to the one shown for the $N\leftrightarrow O$ transition in the main text (Fig.~\ref{fig:br-data}), here shown in red. Note also that for strongly constrained activity, such as $a= 10^1$, the optimal distribution around the \emph{in vivo} value shifts in the opposite way with respect to the unconstrained case (blue against red bar-plots in the lower left of Fig.~\ref{fig:N-O-activity-constraint}). 

Finally, in the upper left of Fig.~\ref{fig:N-O-activity-constraint} we plot the current as a function of $\Delta \psi$. We observe that the actual (baseline-and-control) current effectively goes to zero (stalls) at large potentials, since there is not enough free energy in the cycle to push protons across the membrane. At low potentials, the actual current saturates at a larger value while the harvesting rate plummets. This is because at sufficiently low $\Delta \psi$, transporting protons from inside to outside the cell actually drains free energy stored in the membrane potential, so the cycle operates as a dud.

\emph{Thermodynamic affinity constrained:} Next, we add further constraints to the optimization problem by including thermodynamic affinity limitations, defined by the feasibility set:
\begin{align}
\Lambda_{\mathcal{A},a}^{NO}:=\Lambda_a^{NO} \cap \left\{(\pp,\JJ) : J_{NO}e^{-\mathcal{A}}\leq J_{ON}\leq J_{NO}e^{\mathcal{A}}\right\}.
\end{align}
This linear constraint enforces a bound on the thermodynamic affinity of the $N\leftrightarrow O$ transition,
\begin{align}
    \left|\ln\frac{J_{ON}}{J_{NO}}\right|\leq \mathcal{A} \qquad \text{for some}\quad\mathcal{A} \ge 0\,.
    \label{eq:affmax}
\end{align}

Results for this constraint are shown in Fig.~\ref{fig:N-O-affinity-constraint} with the fixed choice of dynamical activity $a=10^2 \ \text{sec}^{-1}$.
\begin{figure}
\includegraphics[width=.7\columnwidth]{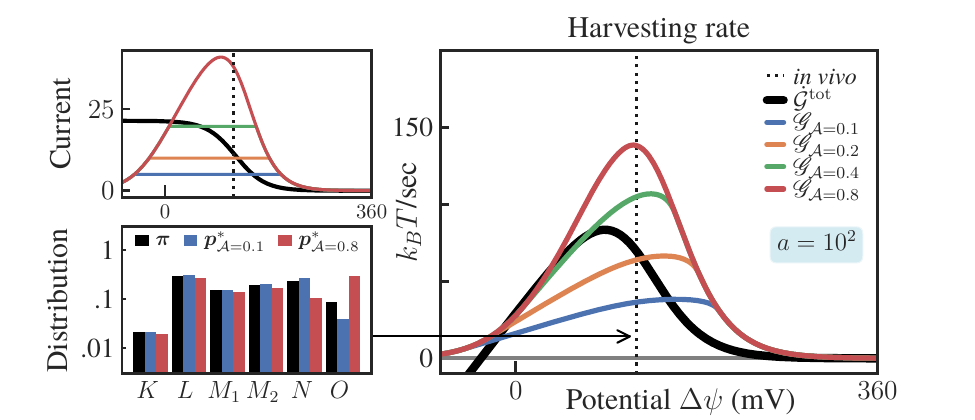}
\caption{\label{fig:N-O-affinity-constraint}
\figurefontcmd
\newstuff{Right: actual harvesting rate (in black) and optimized harvesting rate for controlled transition $N\leftrightarrow O$ at different values of maximum affinity $\mathcal{A}$ and maximum activity $a=10^2\ \text{sec}^{-1}$. Dashed vertical line shows \emph{in vivo} membrane potential ($120\text{mV}$). Left panel shows the respective currents (upper plot) and the distribution at \emph{in vivo} value of $\Delta \psi$ (lower plot) for the actual system and the optimized solutions under $\mathcal{A}=0.1$ and $\mathcal{A}=0.8$.
}
}
\end{figure}

The combination of the constraint on dynamical activity and thermodynamic affinity is necessary to give physically meaningful results. To see why, imagine that only the thermodynamic affinity was constrained, as in \eqref{eq:affmax}.  Now consider some pair of a distribution $\pp$ and flux vector $\JJ$ in our optimization problem \eqref{eq:Gmax00}, where $\JJ$ is restricted to have non-zero transitions only for $N\to O$ and $O\to N$, and satisfy the steady-state condition $\incmatrix \JJ =-\RB \pp$. These fluxes can be increased as $J_{NO}\mapsto J_{NO}+\alpha,J_{ON}\mapsto  J_{ON}+\alpha$ for $\alpha \ge 0$. This transformation doesn't change the current across the transition $N \to O$~\eqref{eq:currentdefinition}, so the steady-state constraint is still valid for the same $\pp$ (note that $\incmatrix \JJ$ only depends on the currents, i.e., net fluxes).  Finally, by choosing $\alpha$ large enough, we can make the EPR term in \eqref{eq:Gmax00} vanish and satisfy the affinity constraint \eqref{eq:affmax}, since $\lim_{\alpha\to\infty}\vert \ln [(J_{NO}+\alpha)/(J_{ON}+\alpha)]\vert =0$. This shows that, lacking other constraints, the affinity constraint is vacuous, since it can be always satisfied by making the one-way fluxes sufficiently large. On the other hand, the combination of the activity and affinity constraints sets a bound on the fluxes.

\subsubsection{Constrained kinetics}

Suppose that we optimize $N\leftrightarrow O$ as control under a constraint involving the respective transition rates. In particular, assume that the possible rates are bounded by some value $\kappa$ such that $R^c_{ON}\le \kappa,R^c_{NO}\leq \kappa$ (in units of $\text{sec}^{-1}$). Then, this can be encoded in the following feasibility set:
\begin{align}\label{eq:activity-constrain1}
\Lambda_{\kappa}^{NO}:=\Lambda^{NO} \cap \{(\pp,\JJ):J_{NO}\leq \kappa p_{O}, J_{ON}\leq \kappa p_{N}\}.
\end{align}
Note that $\Lambda_{\kappa}$ introduces constraints that mix both $\pp$ and $\JJ$. 
We also note that constraints on the transition rates are perhaps the most realistic ones when considering the constraints faced by by biomolecular machines~\cite{miller1990kinetic}. 
Harvesting rate curves for different $\kappa$ are shown in the right plot of Fig.~\ref{fig:N-O-kinetic-constraint}. 

\begin{figure}
\includegraphics[width=.7\columnwidth]{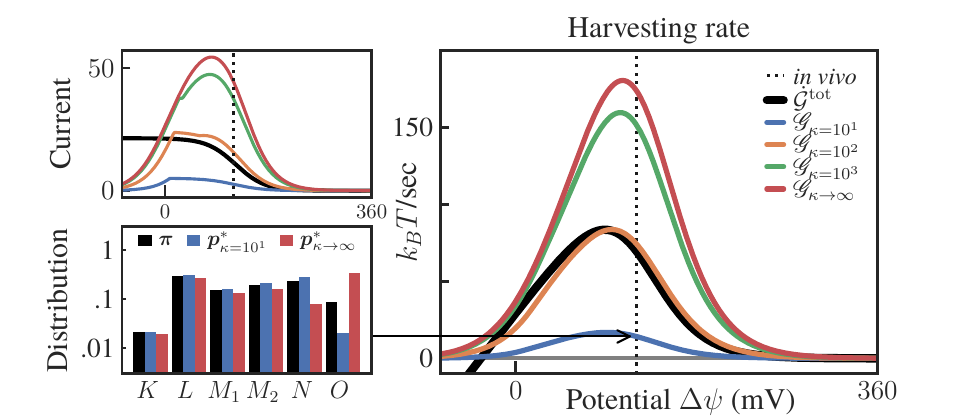}
\caption{\label{fig:N-O-kinetic-constraint}
\figurefontcmd
\newstuff{Right: actual harvesting rate (in black) and optimized harvesting rates for controlled transition $N\leftrightarrow O$ at different values of kinetic constraint $\kappa$. Dashed vertical line shows \emph{in vivo} membrane potential ($120\text{mV}$). Left panel shows the respective currents (upper plot) and the distribution at \emph{in vivo} value of $\Delta \psi$ (lower plot) for the  actual system and the optimized solutions under $\kappa=10^1$ and unconstrained $\kappa$.
}
}
\end{figure}

}

\subsection{Proton-pumping as control}
\label{app:protonpumpascontrol}

The proton-pumping steps ($M_1\leftrightarrow M_2$) is perhaps one of the most experimentally accessible transitions to study. This is because it is the only transition that depends explicitly on the membrane potential $\Delta \psi$, which can be experimentally tweaked and thus used as an external parameter. For example, the electrochemical membrane potential is frequently manipulated by introduction of ionphores, such as valinomycin~\cite{ahmed1983use}.

In this section, we study this transition as control, possibly under constraints.

In this analysis, the baseline rate matrix $\RB$ does not include the transition $M_1 \leftrightarrow M_2$. 
In concrete terms, we let $\RB_{M_2 M_1}=\RB_{M_1 M_2}=0$ and fix the flux parameters of the optimization problem \eqref{eq:Gmax00} as $\JJij=\JJji=0$ for all $\{i,j\}\neq\{M_1,M_2\}$. Formally, we consider the following feasibility set in \eqref{eq:Gmax00}:
\begin{align} % 
\Lambda^{M_1 M_2}:=\{(\pp,\JJ) : \pp \ge 0, \bm{1}^T\pp = 1, \JJ \ge 0, J_{ij}=0 \text{ if } (i,j)\not \in \{{(M_1,M_2)},{(M_2,M_1)}\}\}_{} \,.
\end{align}
along with other constraints discussed below. 

To compare the result of optimization to the actual system, we also consider the actual transition as the control, defined via:
\begin{align}
\RC_{M_2 M_1} = \frac{\kappa_{M_2 M_1}}{1 + e^{-\Delta s^\mathrm{tot}_{M_2 M_1}}} 
\qquad\qquad \RC_{M_1 M_2} = \frac{\kappa_{M_1 M_2}}{1 + e^{-\Delta s^\mathrm{tot}_{M_2 M_1}}} \label{eq:M1M2-as-control}
\end{align}
with all other $\RC_{ij}=0$. The baseline-and-control dynamics correspond to the actual bacteriorhodopsin system described in Sec.~\ref{sec:brmodeldetails}. See also Fig.~\ref{fig:br-sm-baseline-control} (right).

Recall that the proton-pumping transition is the only one that involves the membrane potential $\Delta \psi$ (see~\eqref{eq:appgvals}), and that it is no longer part of the baseline. For this reason, the optimization problem~\eqref{eq:Gmax00}, which only involves baseline parameters, no longer depends on the choice of $\Delta \psi$. We now proceed to study this problem for different choices of additional constraints on $\left(\pp,\JJ\right)$. For this reason, in Fig.~\ref{fig:br-data}, the value for $\Gmax$ when $M_1\leftrightarrow M_2$ acts as control is a constant. In this case, it is interesting to compare the optimum harvesting rates with the values of $\GGdBCbase$ at \emph{in vivo} membrane potential versus $\max_{\Delta \psi} \GGdBCbase$, the maximum value attained by $\GGdBCbase$ across all membrane potentials in Fig.~\ref{fig:br-data} (right). Results are shown in Fig.~\ref{fig:br-sm-proton-pumping}.

\begin{figure}
\includegraphics[height=.205\textheight]{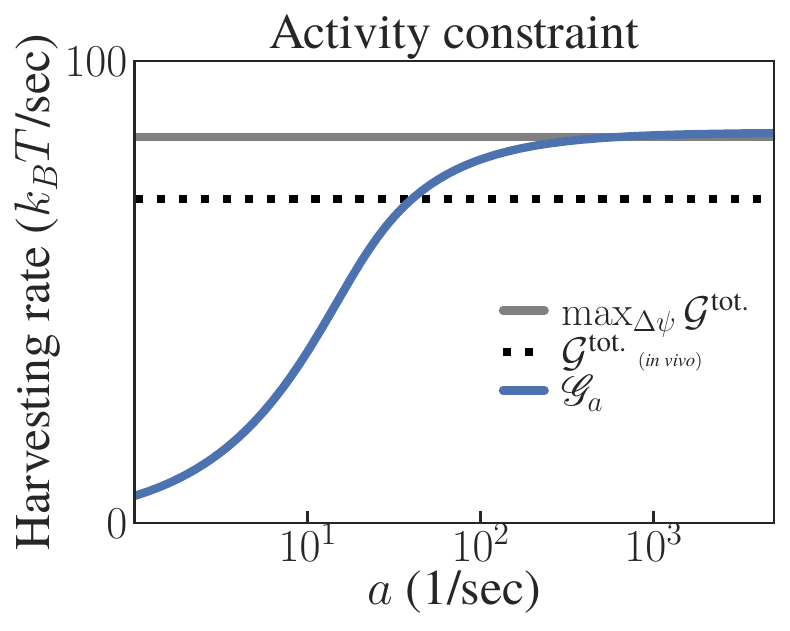}
\includegraphics[height=.205\textheight]{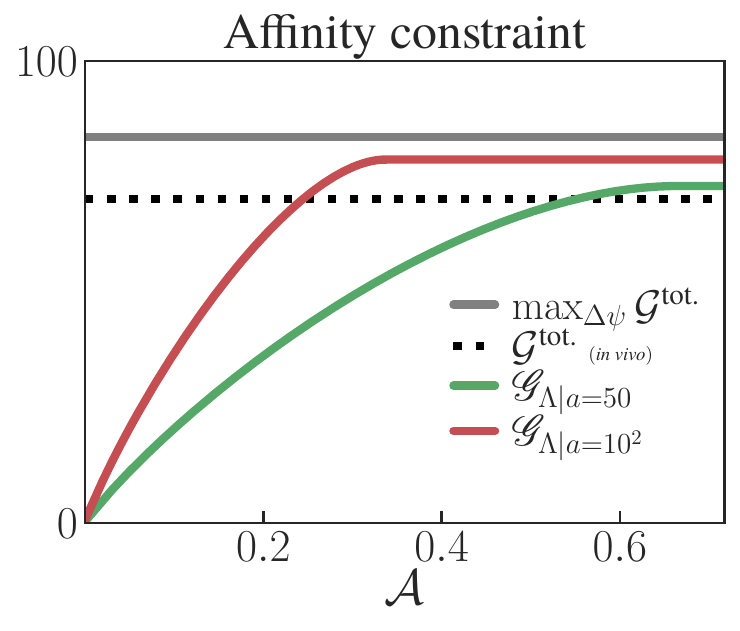}
\includegraphics[height=.205\textheight]{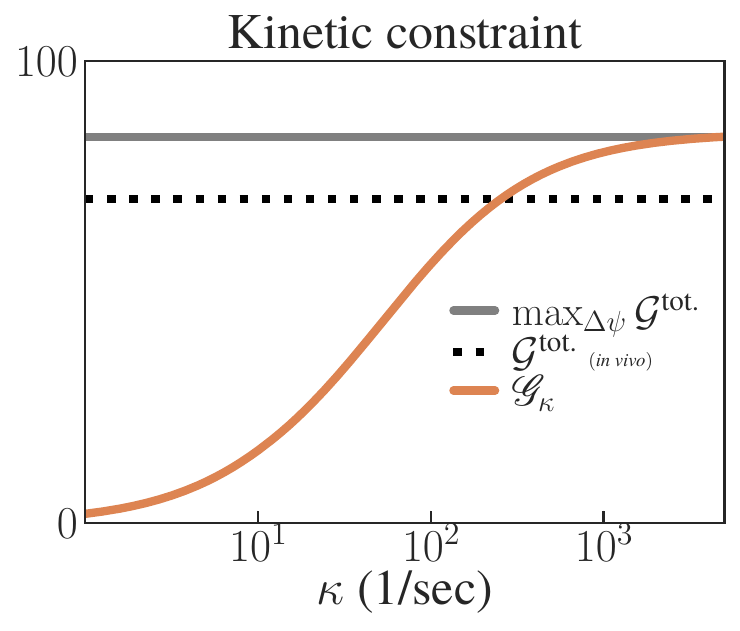}
\makeatletter\long\def\@ifdim#1#2#3{#2}\makeatother
\caption{\label{fig:br-sm-proton-pumping}
\figurefontcmd
Left: Log-log plot of the optimized harvesting rates as a function of dynamical activity bound $a$ ($\Lambda_a^{M_1 M_2}$). Middle: Log-linear plot of the optimized harvesting rates under thermodynamic affinity and dynamical activity constraints ($\Lambda _{\mathcal{A},a}^{M_1 M_2}$) for values of $\mathcal{A}$ at fixed $a=50 \ \text{sec}^{-1}$ (green line) and $a=10^2\  \text{sec}^{-1}$ (red line).  Right: Log-log plot of the optimized harvesting rates as a function of the kinetic bound $\kappa$ ($\Lambda_{\kappa}^{M_1 M_2}$). All figures also show the value of the actual harvesting rate $\GGdBCbase$ at the \emph{in vivo} membrane potential value (black dashed lines), which gives $\sim70k_BT/\text{sec}$, and the maximum achievable harvesting rate that the system can achieve for any $\Delta \psi$ (gray line), $\sim83.5 k_BT/\text{sec}$. }
\end{figure}

\emph{Activity constraint:} analogously to Sec.~\ref{app:reprotonationascontrol}, we now define:
\begin{align}
\Lambda_a^{M_1M_2}:=\Lambda^{M_1M_2}\cap\{(\pp:\JJ) : J_{M_1 M_2}+ J_{M_2 M_1}\leq a, \}\label{eq:act-cons-M1M2}
\end{align}
for values of $a$ in units of $\text{sec}^{-1}$. Solving~\eqref{eq:Gmax00} under $\Lambda_a^{M_1M_2}$ yields the blue curve in Fig.~\ref{fig:br-sm-proton-pumping} for $\Gmax$ as a function of $a$. 

\emph{Thermodynamic affinity constraint:} owing to the same reasoning as in Sec.~\ref{app:reprotonationascontrol}, this constraint proves physically meaningful if combined with another constraint on the fluxes, such as the activity constraint \eqref{eq:act-cons-M1M2}. In this case, we analogously define the feasibility set:
\begin{align}
\Lambda_{\mathcal{A},a}^{M_1M_2}:=\Lambda_a^{M_1M_2}\cap \left\{(\pp,\JJ):J_{M_1 M_2}e^{-\mathcal{A}}\leq J_{M_2 M_1} \le J_{M_1 M_2}e^{\mathcal{A}}\right\}, \label{eq:aff-cons-M1M2}
\end{align}
In the inset of Fig.~\ref{fig:br-sm-proton-pumping}, we show the solution to~\eqref{eq:Gmax00} under $\Lambda_{\mathcal{A},a}^{M_1M_2}$ for values of a range of $\mathcal{A}$ with fixed values of $a=50\ \text{sec}^{-1}$ (green line) and $a=10^2\ \text{sec}^{-1}$ (red line).

\emph{Kinetic constraint:} analogously to Sec.~\ref{app:reprotonationascontrol}, we define:
\begin{align}
\Lambda_{\kappa}^{M_1M_2}:=\Lambda^{M_1M_2}\cap \{(\pp,\JJ) : J_{M_1 M_2}\leq \kappa p_{M_2}, J_{M_2 M_1}\leq \kappa p_{M_1}\}, \label{eq:kin-cons-M1M2}
\end{align}
for values of $\kappa$ in units of $\text{sec}^{-1}$. Solving~\eqref{eq:Gmax00} under $\Lambda_{\kappa}^{M_1M_2}$ yields the orange curve in Fig.~\ref{fig:br-sm-proton-pumping} for $\Gmax$ as a~function~of~$\kappa$. % 

\clearpage
\section{Closed-form solutions of Eq.~(\ref{eq:wmax-1}) in different regimes}
\label{sec:static-env}
%!TeX root = sm.tex

\global\long\def\RS{A}%
\global\long\def\GmaxLR{\mathscr{G}_{\text{LR}}}
\global\long\def\GmaxFE{\mathscr{G}_{\text{ND}}}
\global\long\def\GmaxDF{\mathscr{G}_{\text{D}}}

We consider a system of $\numstates$ states indexed by $i \in\{1,\ldots,\numstates\}$ that evolves according to a Markovian master equation with rate matrix $\RBC=\RB+\RC$, where $\RB$ and $\RC$ correspond to baseline and control processes respectively.  
Here we study the optimization \eqref{eq:wmax-1}  in three limiting regimes.

\subsection{Linear Response (LR) regime, Eq.~\eqref{eq:linear-resp-solution-2}}
\label{subsec:LR-static-env}

We first consider \eqref{eq:wmax-1} 
in the linear response regime, that is under the assumption that the optimal distribution % $p^*$ 
is close to the baseline steady state  $\popt \approx \pstB$, thereby arriving at \eqref{eq:linear-resp-solution-2}. We assume that $\RB$ is irreducible and has a unique steady state $\pstB$ with full support.  
Note that the assumption that $\pstB$ has full support is satisfied when $\RB$ is ``weakly reversible'', meaning that $\RB_{ij} >0$ whenever $\RB_{ji}>0$. 

First, we rewrite the entropic term in  \eqref{eq:wmax-lin-split} as
\begin{align}
\sum_{i,j}\RB_{ij}p_j\ln p_i &= \sum_{i,j}\RB_{ij}p_j\left(\ln p_i-\ln \pstBnb_i\right) +\sum_{i,j}\RB_{ij}p_j\ln \pstBnb_i\label{eq:ent-decomp-LR1}\\
&=\sum_{i,j}\RB_{ij}(p_j-\pstBnb_j)\left(\ln p_i-\ln \pstBnb_i\right)+\sum_{i,j}\RB_{ij}p_j\ln \pstBnb_i, \label{eq:ent-decomp-LR2}
\end{align}
where in the last step we used that $-\sum_{i,j}\RB_{ij}\pstBnb_j\ln \pstBnb_i=0$ when $\pstB$ is the steady-state distribution of the baseline. 
We focus on the first term in \eqref{eq:ent-decomp-LR2}, and use the expansion $\ln x \simeq x-1$ for $x\approx 1$, to derive
\begin{align}
\sum_{i,j}\RB_{ij}(p_j-\pstBnb_j)\ln\left(\frac{p_i}{\pstBnb_i}\right)\approx \sum_{i,j}\RB_{ij}(p_j-\pstBnb_j)\frac{p_i-\pstBnb_i}{\pstBnb_i}=\sum_{i,j}\sqrt{\frac{\pstBnb_j}{\pstBnb_i}}\RB_{ij}z_jz_i,\label{eq:ent-decomp-LR3}
\end{align}
where we defined
\begin{align}
\zz:=D^{-1}(\pp-\pstB)\text{ , with } \  D_{ij}:=\delta_{ij}\sqrt{\pstBnb_i}.\label{eq:z-def}
\end{align}
Our variable to optimize is a probability distribution $\pp$, which needs to satisfy two constrains: $\sum_ip_i=1$ and $p_i\geq 0$. The latter constraint holds automatically, given our hypothesis that $\pp$ is very close to $\pstB$. The former constraint can be expressed in terms of $\zz$ by setting
\begin{align}
\zz^T D^{-1}\pstB = 0 \rightarrow \sum_i\left(p_i-\pstBnb_i\right)=0\iff \sum_i p_i=\sum_i \pstBnb_i=1. \label{eq:z-constrain}
\end{align}
We impose this linear constraint on $\zz$ below, when redefining the optimization problem in terms of $\zz$ rather than $\pp$. 

Before we finish off with the entropic term, we introduce the following symmetric matrix:
\begin{align}
\label{eq:additive-sym}
A_{ij}:=\frac{1}{2}\left(\RB_{ij}+\RB_{ji}\frac{\pstBnb_i}{\pstBnb_j}\right).
\end{align}
The matrix $A$ is sometimes dubbed the ``additive symmetrization'' of $\RB$. It will always obey the detailed balance condition (DB) for $\pstB$,  $\pstBnb_j A_{ij}=\pstBnb_iA_{ji}$. Moreover, $A=\RB$ if and only if $\RB$ obeys DB. We note that only in this latter case would $\pstB$ be an equilibrium distribution. We also define:
\begin{align}
M:=D^{-1}AD\ \rightarrow \ M_{ij}=(D^{-1}AD)_{ij}=\sqrt{\frac{\pstBnb_j}{\pstBnb_i}}A_{ij}.
\end{align}
Observe that $M=D^{-1} A D$, which 
implies that $M$ and $\RS $ are related via a similarity transformation so any right eigenpair $(\lambda,\uu)$ of $\RS$, $\left(\lambda,D^{-1}\uu\right)$ is an eigenpair of the symmetric matrix $M$. Also, since $D$ and $A$ are both symmetric, so is $M$. In particular, since $\RS $ is the sum of two irreducible rate matrices, it is then itself an irreducible rate matrix which has a unique right eigenvector. It is easy to show that, in fact, $\pstB$ is $\RS$'s steady-state distribution with eigenvalue $0$. Then, by similarity, $M$ also has a unique eigenvector $D^{-1}\pstB$ with eigenvalue $0$.    
Moreover, since $A$ is a rate matrix, it is negative semidefinite ($\lambda_\evndx \le 0$), and so $M$ is also negative semidefinite. 
We indicate the right eigenvectors of $A$ as $\uua$ and the eigenvectors of the symmetric matrix $M$ as $\mma=D^{-1}\uua$. We assume that $\mma$ are normalized as $\Vert \mma\Vert =1$, therefore $\Vert D^{-1}\uua\Vert =1$.

Going back to \eqref{eq:ent-decomp-LR3}, we note that 
\begin{align}
\sum_{ij}\sqrt{\frac{\pstBnb_j}{\pstBnb_i}}\RB_{ij}z_jz_i&=\frac{1}{2}\left(\sum_{ij}\sqrt{\frac{\pstBnb_j}{\pstBnb_i}}\RB_{ij}z_jz_i+\sum_{i,j}\sqrt{\frac{\pstBnb_i}{\pstBnb_j}}\RB_{ji}z_iz_j\right)\nonumber \\
&=\sum_{i,j}\sqrt{\frac{\pstBnb_j}{\pstBnb_i}}\left(\RB_{ij}+\frac{\pstBnb_i}{\pstBnb_j}\RB_{ji}\right)z_jz_i=\sum_{i,j}\sqrt{\frac{\pstBnb_j}{\pstBnb_i}}z_jz_i=\sum_{i,j}M_{i,j}z_jz_i.\label {eq:ent-decomp-LR4}
\end{align}

In order to keep track of the second term in  \eqref{eq:ent-decomp-LR2}, let us first recall the free-energy term in \eqref{eq:wmax-lin-split} and the pay-off vector $\hvecbold$ which, in terms of the baseline parameters, reads as $$\hvec_{i}:=\dot{g}^b_{i}+\sum_{j}\RB_{ji}(f_{j}-f_{i}+\gBi_{ji}),$$ and conveniently redefine it such that:
\begin{align}
\psiLRvec= \boldsymbol{\hvec}+\invBeta {\RB}^T\ln\pstB. \label{eq:free-energy-payoff-LR}
\end{align}
For reasons that will become obvious below, we re-scale \eqref{eq:free-energy-payoff-LR} by setting
\begin{align}
\vv:=\frac{\beta}{2}D\psiLRvec. \label{eq:v-def}
\end{align}
We combine the above with the definition \eqref{eq:z-def} and the nonlinear term \eqref{eq:ent-decomp-LR4}. We then recast the optimization problem in terms of the variable $\zz$, which is constrained by \eqref{eq:z-constrain}. Finally, under the LR regime, \eqref{eq:wmax-1} is approximated by
\begin{align}\label{eq:Conditional-Linearized-Target}
\Gmax\simeq\GmaxLR:=\hvecbold^T \pstB + \invBeta \max_{\zz\in \mathbb{R}^n:\zz^TD^{-1}\pstB=0} \zz^TM\zz +2\zz^T\vv.
\end{align}

Eq.~\eqref{eq:Conditional-Linearized-Target} is a quadratic optimization problem, which can be solved using standard techniques from linear algebra. For convenience we summarize these techniques in Lemma~\ref{thm:quad} in Sec.~\ref{sec:quad-theorem} of  this SM. That theorem implies that
\begin{align}
\label{eq:df3qqq}
\GmaxLR =\hvecbold^T \pstB -\invBeta \vv^{T}M^{+}\vv \qquad \qquad \ppLR = \pstB -DM^{+}\vv \,
\end{align}
where  $M^{+}=\sum_{\evndx>1}\lambda_{\evndx}^{-1}\mma{\mma}^T$ is the Moore-Penrose pseudo-inverse of $M$. In applying Theorem~\ref{thm:quad}, we used result \eqref{eq:specialcase2} and the relation $\popt = \pstB + D\zz^*$, where $\zz^*$ solves \eqref{eq:Conditional-Linearized-Target}.

We can rewrite \eqref{eq:df3qqq} using the eigensystem of $M$,
\begin{align}
\GmaxLR = \hvecbold^T \pstB +\invBeta  \sum_{\evndx>1}\frac{ (\vv^T\mma)^2}{-\lambda_{\evndx}} \qquad\qquad  \ppLR=\pstB + \sum_{\evndx>1}\frac{ \vv^T\mma}{-\lambda_{\evndx}}D\mma\,
\end{align} 
Finally, as in the main text we define
\begin{align}
\Omega_\evndx=\invBeta \vv^T\mma=\frac{1}{2} \left( D\psiLRvec\right )^T\left(D^{-1} \uua\right)= \frac{1}{2} \left(\psiLRvec\right)^T \uua,
\label{eq:omdef}
\end{align}
which can be interpreted as the harvesting amplitude of the eigenmodes of $A$. Plugging into \eqref{eq:df3qqq} gives Eq.~\eqref{eq:linear-resp-solution-2} in the main text,
\begin{align}
\label{eq:Prob-Diff-LR}
\GmaxLR = \hvecbold^T \pstB + {\beta} \sum_{\evndx>1}\frac{|\Omega_\evndx|^2}{-\lambda_{\evndx}} \qquad\text{and}\qquad  \ppLR =\pstB + \sum_{\evndx>1}\frac{\beta\Omega_\evndx}{-\lambda_{\evndx}}\uua\,
\end{align}

\subsubsection*{Region of validity of the LR regime}
Note that the LR regime applies when $\Vert \popt-\pstB\Vert \ll 1$. We can write
\begin{align}
\label{eq:LR-validity-RAW}
\begin{aligned}
\Vert \ppLR-\pstB\Vert&= \beta\Big\Vert\sum_{\evndx>1} \frac{\Omega_{\evndx}}{-\lambda_{\evndx}}\uua \Big\Vert \leq \sum_{\evndx>1}\frac{\beta\vert \Omega_{\evndx}\vert}{-\lambda_{\evndx}}\Vert\uua\Vert\\
&\leq \sum_{\evndx>1}\frac{\beta \vert \Omega_{\evndx}\vert}{-\lambda_{\evndx}}\Vert D\Vert\Vert \mma\Vert 
\le
 \beta (n-1) \max_{\evndx>1} \vert \Omega_\evndx/\lambda_\evndx\vert \max_i \sqrt{\pstBnb_i}\,,
\end{aligned}
\end{align}
where we used the triangle inequality and properties of the matrix norm. Therefore, the LR approximation is guaranteed to be valid when 
\begin{align}
\label{eq:LR-validity}
\max_{\evndx>1}\vert \Omega_{\evndx}/\lambda_\evndx\vert \ll \frac{1}{\beta(n-1)\max_i \sqrt{\pstBnb_i}}\,.
\end{align} 
In other words, the harvesting amplitude on each eigenmode must be slower than relaxation modes, up to a factor that depends only on $\beta$, $n$, and the steady state $\pstB$.

%%%%%%%%%%%%%%%%%%%%%%%%%%%%%%%%%%%%%%%%%%%%%%%%%%%%%%%%%%%%%%%%%%%%%%%%%%%%%%%%%%%%%%%%%%%%%%%%%%%
%%%%%%%%%%%%%%%%%%%%%%%%%%%%%%%%%%%%%%%%%%%%%%%%%%%%%%%%%%%%%%%%%%%%%%%%%%%%%%%%%%%%%%%%%%%%%%%%%%%

\subsection{Deterministic (D) regime, Eq.~\eqref{eq:DF}}
\label{subsec:macro-limit}
In the Deterministic (D) regime, we assume that the nonlinear terms in \eqref{eq:wmax-lin-split} are small compared to the linear terms. 
This allows us to approximate the optimal solution of \eqref{eq:wmax-lin-split} as the linear optimization
\begin{align}
\Gmax &=\max_{\pp}\, \invBeta  \left(\RB\pp\right)^T\ln \pp +  \hvecbold^T \pp\label{eq:app333}\\
&\approx \max_{\pp} \hvecbold^T \pp =:\GmaxDF,
\label{eq:appD98}
\end{align}
where $(\ln \pp)^T : = (\ln p_1, \ldots,\ln p_n)$. The maximum in \eqref{eq:appD98} is achieved by $\left(\ppM\right)_i=\delta_{ii^*}$ with $i^*=\arg\max_i \hvec_i$, so
\begin{align}
    \GmaxDF=\phi_{i^*} = \max_{i} \big[\dot{g}^b_{i}+\sum_{j}\RB_{ji}(f_{j}-f_{i}+\gBi_{ji}) \big]\,,
\end{align}
as appears in the main text.
% This appears as \eqref{eq:DF} in the main text. 

\subsubsection*{Region of validity of the D regime}
\label{sec:Mvalidity}

We now discuss when the Deterministic regime approximation is valid. 

We begin by making the assumption that the baseline harvesting rate is non-negative, $\GGd(\pstB) \ge 0$. 
In addition, for notational convenience, let $K:=\max_i  (-\RB_{i i})$ indicate the largest escape rate. 

We express the regime of validity in terms of the following two parameters:
\begin{align}
\alpha &:=\frac{K}{\beta \GmaxDF} = 
\frac{\max_i (-\RB_{i i})}{\beta \max_{i} \big[\ggBi_{i}+\sum_{j}\RB_{ji}(f_{j}-f_{i}+\gBi_{ji}) \big]} \ge 0\label{eq:alpha-def} \\
\gammaConst &:= -\ln\Big(\min_{i} \pstBnb_i \Big) \ge 0\,.
\end{align}
The parameter $\alpha$ reflects the balance between diffusion out of the optimal state versus Deterministic harvesting rate. 
The parameter $\gamma$ is the minimal steady-state probability of any state.   
 
Assuming $\alpha <1$, we show that 
\begin{align}
\label{eq:goal-proof-macrolimit}
\left|\frac{\Gmax}{\GmaxDF}-1\right|\leq \alpha\left( -\ln \alpha + \gammaConst + 1\right)\,
\end{align}
The RHS of~\eqref{eq:goal-proof-macrolimit} vanishes for $\alpha\to0$, so the LHS tightens as $\Gmax/\GmaxDF\to 1$. 
We emphasize that the D approximation only becomes accurate in relative, not additive, terms (i.e., it is not necessarily true that $\Gmax - \GmaxDF\to 0$).

To derive \eqref{eq:goal-proof-macrolimit}, we first consider an upper bound on $\Gmax$. Observe that  
for any $\pp$, 
\begin{align}
-\derivSRB{\pp}=\sum_{i,j} \RB_{ij}p_j \ln p_i \le  \sum_{i} \RB_{ii} p_i\ln p_i \leq \sum_{i}| \RB_{ii}|(-p_i\ln p_i) \leq K S(\pp) \leq K\ln n %\leq K\gammaConst \,,
\label{eq:dtSupperbound}
\end{align}
where we used $\RB_{ij} p_j \ln p_i \le 0$ for $i\ne j$ and $S(\pp)\le \ln n$.    
Plugging into \eqref{eq:app333} and using that $\GmaxDF \ge \hvecbold^T \pp$ for all $\pp$ and $\gammaConst \ge \ln n$, we then have
\begin{align}
\label{eq:upper-bound-deltaWmax-delta}
\Gmax\leq K(\ln n)/\beta + \GmaxDF=  \GmaxDF\alpha\ln n + \GmaxDF \le  \GmaxDF\alpha\gammaConst+ \GmaxDF \,.
\end{align}

To derive a lower bound on $\Gmax$,  define the distribution $p_i:=\left(1-\alpha\right)\delta_{ii^*}+\alpha\pstBnb_i$, with $\alpha$ as in~\eqref{eq:alpha-def}. Plugging this distribution into the objective in~\eqref{eq:app333} yields
\begin{align}
\Gmax&\geq  \invBeta  \left(\RB\pp\right)^T\ln \pp  +  \hvecbold^T \pp=\invBeta(1-\alpha)\sum_i (\RB_{ii^*}\ln p_i) +  \hvecbold^T \pp,
\label{eq:text-dist-delta2}
\end{align}
where we used that $\sum_j \RB_{ij}\left(\left(1-\alpha\right)\delta_{ji^*}+\alpha\pstBnb_j\right)=\RB_{ii^*}(1-\alpha)$ since $\RB\pstB=0$. Using $\RB_{i^*i^*}\ln p_{i^*}\geq 0$, we bound the first term on the right side of~\eqref{eq:text-dist-delta2} as 
\begin{align}
\label{eq:lower-bound-MR1}
\sum_i \RB_{ii^*}\ln p_i&\geq \sum_{i:i\neq i^*}\RB_{ii^*}\ln p_i=\sum_{i:i\neq i^*}\RB_{ii^*}\ln (\alpha \pstBnb_i) \nonumber \\
&\geq \sum_{i:i\neq i* }\RB_{ii^*}(\ln \alpha -\gammaConst )\nonumber \\
&=-\RB_{i^*i^*}(\ln \alpha -\gammaConst) \geq K\left(\ln \alpha - \gammaConst\right)\,,
\end{align}
where we used that $\ln \alpha - \gammaConst <0$ in the last inequality. 
For the second term on the right side of~\eqref{eq:text-dist-delta2}, 
\begin{align}
\label{eq:identity}
 \hvecbold^T \pp=\sum_i p_i\phi_i=(1-\alpha)\sum_i\delta_{ii^*}\hvec_i+\alpha\sum_i\pstBnb_i\hvec_i=(1-\alpha)\GmaxDF+\alpha \GGd(\pstB).\end{align}
We can plug~\eqref{eq:lower-bound-MR1} and~\eqref{eq:identity} into~\eqref{eq:text-dist-delta2} to give a lower bound on $\Gmax$, 
\begin{align}
\Gmax&\geq \invBeta (1-\alpha) K (\ln \alpha -\gammaConst)+(1-\alpha)\GmaxDF+ \alpha \GGd(\pstB)\nonumber \\
&= (1-\alpha) \GmaxDF \alpha(\ln \alpha - \gammaConst - 1)+\GmaxDF+ \alpha \GGd(\pstB)\nonumber 
\end{align}
where we used the definition of $\alpha$~\eqref{eq:alpha-def} and rearranged. Since $\ln \alpha -\gammaConst <0$, we can further drop the $\alpha^2$ term to give 
\begin{align}
\Gmax&\ge  \GmaxDF \alpha(\ln \alpha -\gammaConst - 1)+\GmaxDF+ \alpha \GGd(\pstB)\nonumber \,,\\
&\ge  \GmaxDF \alpha(\ln \alpha -\gammaConst - 1)+ \GmaxDF\,,
\label{eq:lower-bound-deltaWmax-delta}
\end{align}
where in the second line we used the assumption $\GGd(\pstB) \ge 0$.
Combining and using that $\alpha <1$ yields
\begin{align}
\label{eq:sandwich-bounds-delta}
\GmaxDF\alpha\left(\ln\alpha -\gammaConst - 1\right)\leq \Gmax-\GmaxDF\leq \GmaxDF \alpha\gammaConst \le  \GmaxDF \alpha(-\ln \alpha + \gammaConst+1)\,,
\end{align}
which can be rearranged to give~\eqref{eq:goal-proof-macrolimit}.

Interestingly, we can also derive a relative perturbation bound on the maximum \emph{increase} of the harvesting rate, above the baseline harvesting rate. Specifically, using a similar derivation as above, we can show that
\begin{align}
    \left|\frac{\Gmax - \GGd(\pstB)}{\GmaxDF - \GGd(\pstB)}-1\right|\leq \alpha\left( -\ln \alpha + \gammaConst + 1\right)\,,
    \label{eq:appPertDbound2}
\end{align}
where $\alpha=K/\beta (\GmaxDF - \GGd(\pstB))\ge 0$ and $\gammaConst$ is defined as before. In fact, this perturbation bound on the increase holds without any additional assumptions, i.e., regardless of the sign of $\GGd(\pstB)$.

%%%%%%%%%%%%%%%%%%%%%%%%%%%%%%%%%%%%%%%%%%%%%%%%%%%%%%%%%%%%%%%%%%%%%%%%%%%%%%%%%%%%%%%%%%%%%%%%%%%
%%%%%%%%%%%%%%%%%%%%%%%%%%%%%%%%%%%%%%%%%%%%%%%%%%%%%%%%%%%%%%%%%%%%%%%%%%%%%%%%%%%%%%%%%%%%%%%%%%%

\subsection{Near-Deterministic (ND) regime, Eq.~\eqref{eq:wdPD}}
\label{subsec:FE-static-env}

In this regime, we consider a perturbation of the Deterministic regime: 
we assume that the optimal $\popt$ is close to a delta-function distribution centered at 
$i^*={\text{argmax}}_i \hvec_i$. 
If $p_i\approx \delta_{ii^*}$,
\begin{align}
\dot{p}_i=\sum_j\RB_{ij}p_j\approx \RB_{ii^*} \,.
\end{align}
This allows us to approximate the entropic term in \eqref{eq:app333} as
\begin{align}
(\RB\pp)^T\ln \pp&=\sum_{i,j}\RB_{ij}p_j\ln p_i\approx \sum_i \RB_{ii^*}\ln p_i \nonumber \\
&= \RB_{i^*i^*}\ln\Big(1-\sum_{i:i\neq i^*} p_i\Big)+\sum_{i:i\neq i^*} \RB_{ii^*}\ln p_i\nonumber \\
&\approx\sum_{i:i\neq i^*} \left(-\RB_{i^*i^*}p_i+\RB_{ii^*}\ln p_i\right)
\end{align}
where we used $\ln (1-x) \approx -x$ when $x\approx 0$. Plugging into~\eqref{eq:app333} gives
\begin{align}
\label{eq:PD-Target-Function}
\Gmax&\simeq\max_{\pp} \invBeta \sum_{i:i\neq i^*}\left(-\RB_{i^*i^*}p_i+\RB_{ii^*}\ln p_i\right)+\Big(1-\sum_{i:i\neq i^*} p_i\Big)\hvec_{i^*}+\sum_{i:i\neq i^*} p_i\hvec_i\nonumber \\
&=\GmaxDF+\max_{\pp}\sum_{i:i\neq i^*}\left[-p_i\left(\hvec_{i^*}-\hvec_i+\beta^{-1}\RB_{i^*i^*}\right)+\beta^{-1}\RB_{ii^*}\ln p_i\right]=:\GmaxFE\,,
\end{align}
where we recall that $\GmaxDF=\hvec_{i^*}$, which does not depend on $\pp$. After maximizing with respect to $\{p_i\}_{i:i\neq i^*}$ by taking derivatives and setting them to zero, one obtains
\begin{align}
\label{eq:Optimum-p-PD}
(\ppFE)_{i}=\frac{\RB_{ii^*}}{\beta(\hvec_{i^*}-\hvec_i)+ \RB_{i^*i^*}}\qquad (\text{for}\,i\neq i^*)\qquad  {\text{and}}\qquad\, (\ppFE)_{i^*}=1-\sum_{i:i\neq i^*} (\ppFE)_i\,.
\end{align}
Plugging~\eqref{eq:Optimum-p-PD} into~\eqref{eq:PD-Target-Function} gives
\begin{align}
\label{eq:Optimium-W-PD}
\GmaxFE = \GmaxDF + \invBeta \sum_{i:i\neq i^*} \RB_{ii^*}\left[\ln(\ppFE)_i-1\right]\,.
\end{align}

\subsubsection*{Region of validity of the ND regime}

% At the very least, all else being equal, it should not decrease in $\hvec_{i^*}$. We can consider the derivative 
% \begin{align*}
% \frac{\partial}{\partial \hvec_{i^*}} \GmaxFE  &= 1 - \frac{\partial}{\partial \hvec_{i^*}}\invBeta\sum_{i:i\neq i^*} \RB_{ii^*}\ln [\beta(\hvec_{i^*}-\hvec_i)+ \RB_{i^*i^*}]\\
%  &= 1 - \sum_{i:i\neq i^*}\frac{ \RB_{ii^*}}{\beta(\hvec_{i^*}-\hvec_i)+ \RB_{i^*i^*}}
% \end{align*}

The ND approximation applies when $(\ppFE)_{i^*} \approx 1$, or equivalently when
$$
1\gg\sum_{i:i\neq i^*} (\ppFE)_{i} = \sum_{i:i\neq i^*} \frac{\RB_{ii^*}}{\beta(\hvec_{i^*}-\hvec_i)+ \RB_{i^*i^*}}\,.
$$
%(\ppFE)_{i} \approx 1$, which holds when
% $$
% \beta(\hvec_{i^*}-\hvec_i)> -\RB_{i^*i^*} \qquad\text{for all} \quad i.
% $$
% \end{align}
% It is also necessary $(\ppFE)_{i^*}\approx 0$, or in other words that 
% \begin{align}
% 1\gg \sum_{i:i\neq i^*} p_i^* = \sum_{i:i\neq i^*}\frac{\RB_{ii^*}}{\beta(\hvec_{i^*}-\hvec_i)+ \RB_{i^*i^*}}\,.
% \end{align}
For convenience, %let $i^-=\text{arg}\max_{i\neq i^*} \hvec_i$ indicate the second most valuable state of the system after $i^*$. We 
we define
the {\em harvesting gap} as the difference in the pay-off between the optimal state $i^*$ and the second optimal state,
$$
\Delta \psiMmin:= %\hvec_{i^*}-\hvec_{i^-}=
\hvec_{i^*} - \max_{i:i\neq i^*} \hvec_i\,.
$$
By combining and rearranging the above, we can simplify the validity condition for the ND approximation as
\begin{align}
\label{eq:FE-validity}
\beta \Delta \psiMmin \gg  -2\RB_{i^*i^*}, 
\end{align}
% % Let us briefly discuss the validity of the expressions above. As stated above,~\eqref{eq:Optimum-p-PD} is valid when $p_{i}\approx \delta_{ii^*}$. For convenience, let $i^-=\text{arg}\max_{i\neq i^*} \hvec_i$ indicate the second most valuable state of the system after $i^*$. We define
% % the {\em harvesting gap} as 
% % $$
% % \Delta \psiMmin:= \hvec_{i^*}-\hvec_{i^-}=\min_{i:i\neq i^*} \hvec_{i^*}-\hvec_i\,.
% % $$
% % Now, considering the optimal probability for $i\neq i^*$ found in~\eqref{eq:Optimum-p-PD},  
% A simple the ND approximation will be valid
% as long as
% \begin{align*}
% p^*_{i\neq  i^*}\ll 1 \quad \Leftrightarrow \quad \hvec_{i^*}-\hvec_i\gg \invBeta  \left(\RB_{ii^*}-\RB_{i^*i^*}\right)\,,
% \end{align*}
% which holds as long as
% \begin{align}
% \label{eq:FE-validity}
% \Delta \psiMmin\gg \invBeta \left(\RB_{i^-i^*}-\RB_{i^*i^*}\right)\,.
% \end{align}
Note that this condition also guarantees that $(\ppFE)_{i}\ge 0$ for $i\ne i^*$. 
Expression~\eqref{eq:FE-validity} implies that ND approximation is valid when the harvesting gap is much larger than the rates of escape from the optimal state $i^*$. 
%at which either $i^*$ flows into $i^-$ or escapes.

\clearpage
\section{Unicyclic model}
\label{sec:rings1}
% !TeX root = sm.tex

\global\long\def\cycla{\omega_a}
\global\long\def\cyclb{\omega_b}
\global\long\def\cyclc{\omega_c}
\global\long\def\cyclabo{\omega_{a+b-1}}
\global\long\def\Reff{\mathcal{R}}
\global\long\def\GmaxLR{\mathscr{G}_{\text{LR}}}
\global\long\def\GmaxFE{\mathscr{G}_{\text{ND}}}
\global\long\def\GmaxDF{\mathscr{G}_{\text{D}}}

Here we analyze the unicyclic model considered in the main text. Before proceeding, we briefly introduce some useful facts about the eigendecomposition of a unicylic rate matrix.

\subsection{An Algebraic Aperitif: eigendecomposition of unicyclic rate matrix}
\label{sec:algebraic-aperitif}

Consider a unicylic rate matrix $\RB$,
\begin{align}
\label{eq:Ring_Sates_Rate_Matrix_Baseline}
\RB= \left(\begin{matrix}
-2k & k & 0 & \cdots & k \\
k & -2k & k & \cdots & 0 \\
0 & k & -2k & \cdots & 0 \\
\vdots & \vdots & \vdots & \ddots & \vdots \\
k & 0 & 0 & \cdots & -2k 
\end{matrix}\right)\,,
\end{align}
which is a symmetric circulant matrix. Due to symmetry, its steady state is uniform: $\pstBnb_i=1/n$ for all $i$. The eigensystem for $R$ is obtained from the theory of circulant matrices~\cite{gray2006toeplitz}, which, for odd $n$, yields:
\begin{align}
\lambda_{a}&=-2k\left[1-\Re\left(\cycla\right)\right]=-2k\left[1-\cos\left(\frac{2\pi(a-1)}{n}\right)\right] \,,
\label{eq:eigLamExp}
\end{align}
where $\cycla:=e^{\mathrm{i}2\pi (a -1)  /n}$. These eigenvalues are all degenerate twice (except $\lambda_{(1)}=0$, whose eigenvector is simply $(1,1,\ldots,1)$). An orthonormal choice of eigenbasis is given by the set 
\begin{align}\label{eq:eigenvectors-real-basis}
\{m_{(a)}\}=\left\{\frac{1}{\sqrt{n}}\left(1,\cycla,\cycla^2,\ldots,\cycla^{n-1}\right)\right\}.
\end{align}

\subsection{Transition harvesting cycle}
\label{sec:uni-static}

We now analyze the unicyclic model using techniques developed in Sec.~\ref{sec:static-env}. 
We consider a systems with $n$ state arranged in a ring, where baseline transitions are symmetric with uniform kinetics: $i\overset{k}{\underset{k}{\rightleftharpoons}}i+1\mod n$, $\forall i=1,\ldots,n$ (see Fig.~\ref{fig:rings} here and Fig.~\ref{fig:two-rings}~(a) in the main text). The baseline dynamics are equivalent to a random walk on a one-dimensional ring, and the baseline rate matrix is given by \eqref{eq:Ring_Sates_Rate_Matrix_Baseline}.  

We assume that a single transition, taken to be $1 \to 2$ without loss of generality, harvests $\gBi_{21}=-\gBi_{12}=\Theta$ joules of free energy. 
In addition, note that in the main text we chose to set the units of $k=1$ without loss of generality. However, 
we have kept $k$ explicit in the rest of the SM.  For this system, the baseline steady state is uniform, $\pstBnb_i=1/n$, and the harvesting vector is
\begin{align}
\label{eq:payoff3}
\psiLRvec=\hvecbold=(k\Theta,-k\Theta,0,\ldots,0)\,.
\end{align}

\begin{SCfigure}
    \centering
\includegraphics[trim={0 0 12cm 0},clip, width=0.3\textwidth]{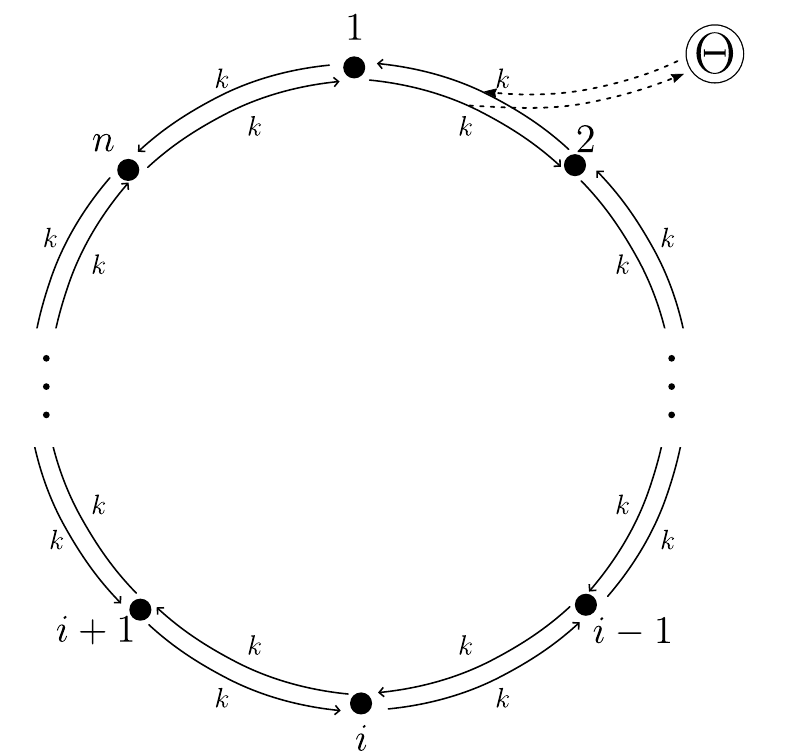}
\caption{\label{fig:rings}
\figurefontcmd
A unicyclic system across the $n$ states with symmetric back-and-forth rates $k$, as discussed in the main text. An amount of $\Theta$ joules is harvested as free energy into the internal reservoir for each transition $1\to 2$ (and vice versa for $2\to 1$). 
}
\end{SCfigure}

In the rest of this section, we analyze this model in the LR, D and ND regimes. 
We will also discuss the validity of these approximation in terms of the model parameters.

\subsubsection{LR regime}

We first consider the LR solution~\eqref{eq:Prob-Diff-LR}. 
Combining \eqref{eq:payoff3} with~\eqref{eq:v-def}, $\pstBnb_i=1/n$, and $D=\frac{1}{\sqrt{n}}I_n$ gives $\vv=\frac{k\beta\Theta}{2\sqrt{n}}(1,-1,0,\ldots,0)$.  
In addition, since in this particular example $R$ obeys detailed balance, $\RB=\RS $. Finally, it is easy to verify that $M=D^{-1}\RS D=\RS =\RB$. Thus, $M$ has the same eigendecomposition as $R$. The normalized eigenvectors of $M$ are given by \eqref{eq:eigenvectors-real-basis}. 

We can compute $\Omega_{a}=\invBeta \vv^T\mma$ for $a=2,\ldots,n$, using the eigenvector set given in~\eqref{eq:eigenvectors-real-basis},which simply yields 
\begin{align}
\Omega_{a}=\frac{k \Theta}{2n}\left(1-\cycla\right)\,.
\label{eq:omegaAexp}
\end{align}

Next, we compute the maximum harvesting rate attainable using~\eqref{eq:Prob-Diff-LR},
\begin{align}
\label{eq:deltaWmax-stat-ring}
 \Gmax_{\text{LR}} = \sum_{a=2}^{n}\frac{\beta\left|\Omega_a\right|^{2}}{-\lambda_{a}} &=
\beta k \left(\frac{\Theta}{2n}\right)^2  \sum_{a=2}^{n}\frac{\left|1-\cycla\right|^2}{2\left[1-\cos\left(\frac{2\pi(a-1)}{n}\right)\right]} \nonumber \\
&=\beta k \left(\frac{\Theta}{2n}\right)^2  \sum_{a=2}^{n}1\nonumber\\
&=\beta k \left(\frac{\Theta}{2n}\right)^2(n-1)\,.
\end{align}

On the other hand, continuing with odd $n$, it is also possible to compute the deviation of the optimal distribution from stationary distribution using~\eqref{eq:Prob-Diff-LR},
\begin{align}
\label{eq:prob-diff-trans}
\Delta \ppLR :=
\ppLR-\pstB=\beta\sum_{a=2}^{n} \frac{\Omega_{a}^{\dagger}D\mma}{-\lambda_{a}}\,,
\end{align}
which is depicted in Fig.~\ref{fig:delta_ps}(a). 
(Note that, due to our choice of a complex-valued eigenbasis, we must be careful in adding the complex conjugate on the $\Omega_a$, following the ordering of the operator $M^+$ in~\eqref{eq:Prob-Diff-LR}.)
Now, component by component, we can rewrite~\eqref{eq:prob-diff-trans} as
\begin{align}
\left(\Delta \ppLR\right)_i &=\frac{\beta\pwr}{4n^2} \sum_{a=2}^{n} \frac{(1-\cycla)^\dagger\cycla^{i-1}}{1-\cos\left(\frac{2\pi(a-1)}{n}\right)} \nonumber \\
&=\frac{\beta\pwr}{4n^2} \sum_{a=2}^{n} \frac{\cos\left(\frac{2\pi(a-1)(i-1)}{n}\right)-\cos\left(\frac{2\pi(a-1)(i-2)}{n}\right)}{1-\cos\left(\frac{2\pi(a-1)}{n}\right)}\nonumber \\
&=\frac{\beta\pwr}{4n^2}  \left[-\sum_{a=2}^{n}\cos\left(\frac{2\pi(a-1)(i-2)}{n}\right)-\sum_{a=2}^{n}\frac{\sin\left(\frac{2\pi(a-1)(i-2)}{n}\right)\sin\left(\frac{2\pi(a-1)}{n}\right)}{1-\cos\left(\frac{2\pi(a-1)}{n}\right)}\right].
\label{eq:prob-diff-trans-components0}
\end{align}
In the first line, we used the expressions for $\Omega_a$ from \eqref{eq:omegaAexp}, $\lambda_a$ from \eqref{eq:eigLamExp}, $\mma$ from \eqref{eq:eigenvectors-real-basis}, and $D= I/\sqrt{n}$ and then simplified. In the second line, we expanded $(1-\cycla)^\dagger\cycla^{i-1}=\cycla^{i-1}-\cycla^{i-2}$ into real and imaginary components and then used that the imaginary components cancel over the sum. In the last line, we used the identity $\cos(x+y)=\cos(x)\cos(y)-\sin(x)\sin(y)$ for $x=\frac{2\pi(a-1)(i-2)}{n}$ and $y=\frac{2\pi(a-1)}{n}$ and then simplified. It can be verified that the first series in \eqref{eq:prob-diff-trans-components0} sums to $-1$. The second series is trickier but can be simplified using trigonometric methods as shown in Ref.~\cite{MSE104967}. Plugging in that solution and simplifying gives the very simple expression:
\begin{align}
\label{eq:prob-diff-trans-components}
\left(\Delta \ppLR\right)_i  =\frac{\beta  \Theta}{4n^2}\begin{cases}2(i-1)-(n+1) &{\text{\qquad for }} i=2,\ldots,n \\
(n-1)  &{\text{\qquad for }} i=1\end{cases}\,.
\end{align}
Thus, in the LR regime, the optimal distribution builds up in equal increments starting at $i=2$ until the optimal state $i=1$, after which it falls off a cliff. This is shown in Fig.~\ref{fig:delta_ps}.

\begin{SCfigure}
\centering{
\includegraphics[trim={0 0 18cm 0},clip, width=0.4\textwidth]{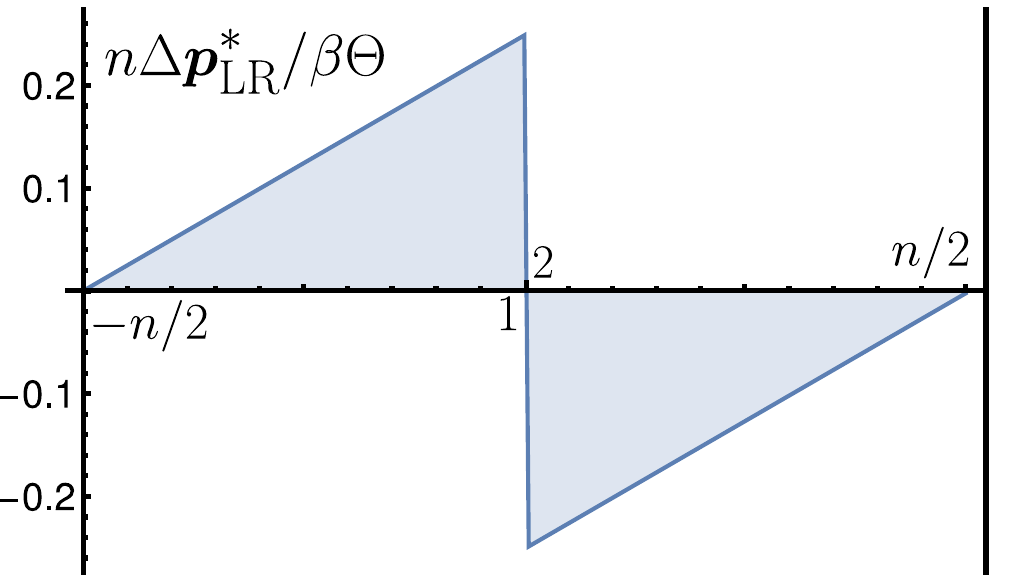} 
\caption{\label{fig:delta_ps} 
\figurefontcmd
The optimal $\Delta \ppLR$ (deviation of optimal distribution from the uniform baseline steady state)  given~\eqref{eq:prob-diff-trans-components}. 
The optimal distribution exhibits a clockwise cyclic current given by~\eqref{eq:Global-flux-irreversible-ring}, 
%; while the latter does not have any cyclic current. 
The magnitude of $\Delta \ppLR$  decays as $\sim n^{-1}$, so larger rings will result in optimal distributions that are closer to the baseline steady state.
}}
\end{SCfigure}

We can also compute the probability current across the transition $1\to 2$, which we leave as an exercise for the reader:
\begin{align}
\label{eq:Global-flux-irreversible-ring}
k{\left[(\ppLR)_1-(\ppLR)_2\right]}= \frac{k\beta\Theta}{2n}\left(\frac{n-1 }{n}\right)\,.
\end{align}

The expressions~\eqref{eq:LR-validity-RAW} and~\eqref{eq:LR-validity} suggest for which values of $\Theta$ the optimal solution will be in the LR regime. 
In this case, however, it is possible to work out the norm $\Vert \ppLR-\pstB\Vert $ exactly:
$$ \Vert \Delta \ppLR \Vert  = \frac{\beta\Theta}{4n^2} \left[\sum_{x=2}^n\left(2x-n-3\right)^2+(n-1)^2\right]^{1/2}=\frac{\beta \Theta}{4\sqrt{3}}\sqrt{\frac{n^2-1}{n^3}}$$
which implies that the LR regime is valid when
\begin{align}
\Vert \ppLR-\pstB \Vert \ll 1 \Leftrightarrow \Theta \ll \frac{4\sqrt{3}}{\beta}\sqrt{\frac{n^3}{n^2-1}}
\end{align}

In this case, the larger the number of states $n$ in the ring, the wider the interval of $\Theta$ values that will make the optimal solution fall into the LR regime. For large $n$, we obtain $$\Theta \ll \frac{4\sqrt{3}}{\beta}\sqrt{n}.$$

\subsubsection{D and ND Regimes}

Note that the optimal mesostate is $i^*=1$. Hence $\GmaxDF=\hvec_1=k\Theta$ and $(\ppM)_i=\delta_{i1}$. Next, consider the ND expressions~\eqref{eq:Optimum-p-PD} with the payoff vector \eqref{eq:payoff3}. Denote $\Delta \hvec_2=\hvec_{i^*}-\hvec_2=2k\Theta$ and $\Delta \hvec_n=\hvec_{i^*}-\hvec_n=k\Theta$; and observe that $\RB_{21}=\RB_{n1}=k$, $\RB_{11}=-2k$ and $\RB_{j 1}=0$ for all $j\neq 2,n$. Hence, for $i\neq i^*$, we obtain
\begin{align}
\begin{aligned}
(\ppFE)_{i}&=\frac{\RB_{i 1}}{\beta \Delta \hvec_{i}+\RB_{11}}=\begin{cases}\frac{1}{2(\beta\Theta -1)} &\qquad i=2\\
\frac{1}{\beta \Theta - 2} &\qquad i=n \\
0&\qquad {\text{otherwise}}\end{cases}\\
\left(\ppFE\right)_{i^*=1}&=1-\sum_{i:i\neq i^*}(\ppFE)_{i}=1-\frac{3\beta \Theta-4}{2(\beta\Theta -1)(\beta\Theta -2)} \,.
\end{aligned}
\label{eq:PD-Probability-Deviation-Irrev-Ring}
\end{align}
The harvesting rate attainable in this regime can be calculated using Eq.~\eqref{eq:Optimium-W-PD}, % Eq.~\eqref{eq:wdPD},
\begin{align}
\label{eq:PD-Extra-Power-Irrev-Ring}
\GmaxFE &=\hvec_1+{\beta^{-1}}\sum_{i:i\neq i^*}\RB_{i 1}\left[\ln \left(\ppFE\right)_i-1\right]\\
&=k(\Theta-2/\beta)-k\beta^{-1}\ln\left[2(\beta \Theta-1)(\beta\Theta -2)\right].
\end{align}

\clearpage
\subsection{State harvesting cycle}
\label{sec:rings2}
% !TeX root = sm.tex
\global\long\def\GmaxLR{\mathscr{G}_{\text{LR}}}
\global\long\def\GmaxFE{\mathscr{G}_{\text{ND}}}
\global\long\def\GmaxDF{\mathscr{G}_{\text{D}}}

In this section, we consider the second unicyclic model, see Fig.~\ref{fig:rings-v2} and Fig.~\ref{fig:two-rings}~(b) in the main text. In this model, the baseline dynamics correspond once again to the
unicycle with symmetric rates, as in Sec.~\ref{sec:uni-static}.  However, 
harvesting is not coupled to transitions, but rather to internal fluxes within a single mesostate. Without loss of generality, we choose this mesostate to be $i=1$. \newstuff{In other words, we let $\dot{g}^b_1=\theta>0$ and $\dot{g}^b_{1<i\leq n}=0$.} 
There is now no preferential direction for the  probability current. The free energy harvested per unit time when the system is in mesostate $1$ is given by the parameter $\theta$, which carries the same units as $\Gmax$.
Just as before, the baseline steady state is uniform, i.e. $\pstBnb_i=1/n$. Hence, the harvesting vector is
\begin{align}
\label{eq:payoff4}
\psiLRvec=\hvecbold=(\theta,0,\ldots,0)\,.
\end{align}

%This section provides a complementary illustration of how our theoretical framework unfolds in a different scenario.
% Note that this model is not discussed in the main text, and only appears here in the SM.  
Below we consider the LR, D and ND regimes.

\begin{SCfigure}
\centering
\includegraphics[trim={13cm 0 0 0},clip, width=0.3\textwidth]{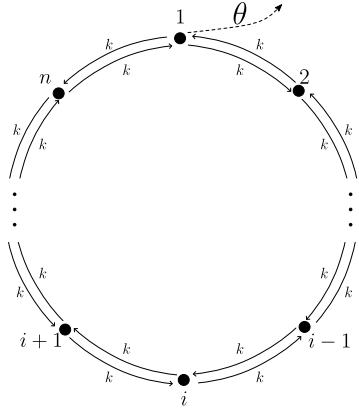}
\caption{\label{fig:rings-v2}
\figurefontcmd
We consider a unicyclic system with  transitions across the $n$ states with symmetric back-and-forth rates $k$. A single state (labelled $i=1$) leads to a flux of free energy at rate $\theta$. 
}
\end{SCfigure}

\subsubsection{LR Regime}

In analogy with the previous example, the expression for the upper bound in free energy harvesting rate under the LR regime here is obtained as
\begin{align}
\GmaxLR =\beta\sum_{a=2}^{n}\frac{\Omega_{a}^{2}}{-\lambda_{a}} \,,
\end{align}
with $\Omega_{a}= \invBeta \vv^T\mma $, where the vectors $\mma$ correspond to the eigensystem discussed in discussed in Sec.~\ref{sec:algebraic-aperitif}, i.e. in~\eqref{eq:eigenvectors-real-basis}. On the other hand, here $\vv=\frac{\beta\theta}{2\sqrt{n}}(1,0,\ldots,0)$. The latter is obtained by construction using~\eqref{eq:z-def}, \eqref{eq:free-energy-payoff-LR} and~\eqref{eq:v-def}. Thus, 
$$\Omega_{a}= \frac{\theta}{2n} \,.$$ 
For large even $n$, it is possible to show that~\cite{MSE4567421}
\begin{align}
\sum_{a=2}^{n/2}\frac{1}{1-\cos\left(\frac{2\pi(a-1)}{n}\right)}\simeq \frac{n^2}{12} + \mathcal{O}(1)\,,
\end{align}
which allows us to write 
\begin{align}\label{eq:deltaWmax-Localized-large-n-LR}
\GmaxLR \simeq\frac{\beta}{12k}\left(\frac{\theta}{2}\right)^2\,.
\end{align}
If we wanted to approximate the result above for a large odd $n$, we would need to add a term proportional to $1/4n^2$, which is of second order, hence the leading order is still captured by expression~\eqref{eq:deltaWmax-Localized-large-n-LR}.

We can also study the optimal distribution using~\eqref{eq:Prob-Diff-LR},
\begin{align}\label{eq:prob-diff-loc}
\Delta \ppLR=\ppLR-\pstB={\beta}\sum_{a=2}^{n} \frac{\Omega_{a}^{\dagger}D\mma}{-\lambda_{a}}=\frac{\beta\theta}{4k\numstates^2}\sum_{a=2}^{n}\frac{\uua}{1-\cos\left(\frac{2\pi(a-1)}{\numstates}\right)}
\end{align}
with $u_{(a)}:=\left(1,w_a,w_a^2,\ldots,w_a^{n-1}\right)$.
This is computed numerically and its behavior is shown in Fig.~\ref{fig:delta_ps}(b). We leave this computation as an exercise to the reader.

\begin{SCfigure}
\centering{
\includegraphics[trim={18cm 0 0 0},clip, width=0.4\textwidth]{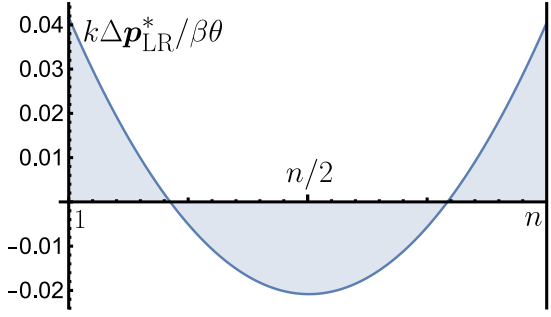} 
\caption{\label{fig:delta_ps-v2} 
\figurefontcmd
The optimal $\Delta \ppLR$ (the deviation of the optimal distribution from the uniform baseline steady state)  from~\eqref{eq:prob-diff-loc}. 
Note that there optimal distribution does not lead to a cyclic current. 
}}
\end{SCfigure}

Expression~\eqref{eq:LR-validity} suggests sufficiency conditions for when the optimum will be in the LR regime. 
Note that, in this case, $\max_{a>1}\Omega_{a}={\theta}/{2n}$ while $$\min_{a> 1}-\lambda_{a}=4k\min_{a> 1}\sin^2\left(\frac{\pi(a-1)}{n}\right)=4k\sin^2\left(\frac{\pi}{n}\right)\,.$$
Then, using $\Vert D\Vert =1/\sqrt{n}$, condition~\eqref{eq:LR-validity} will read as:
\begin{align}
\frac{\theta}{2n}\ll \frac{4k\sqrt{n}}{\beta(n-1)}\sin^2\left(\frac{\pi}{n}\right)\Leftrightarrow \theta \ll \frac{8\sqrt{n}k}{\beta}\left(\frac{n}{n-1}\right)\sin^2\left(\frac{\pi}{n}\right)\,.
\end{align}
For sufficiently large $n$, we can approximate $\sin^2(\frac{\pi}{n})\approx(\frac{\pi}{n})^2$, so the LR regime is valid when
\begin{align}
\theta \ll \frac{8\pi^2}{n^{3/2}}\frac{k}{\beta}\,.
\end{align}
Expression~\eqref{eq:LR-validity}, however, is not necessarily a tight bound. That is, the result hereby obtained is a sufficient condition for the LR regime to be valid, but not a necessary one. 

\subsubsection{D and ND Regimes}

The optimal mesostate is $i^*=1$. Hence 
$$\GmaxDF=\hvec_1=\theta \qquad \qquad (\ppM)_i=\delta_{i1}\,.$$ 
Next, by following a similar procedure as above, we denote $\Delta \hvec_2=\Delta \hvec_n=\theta$, $\RB_{21}=\RB_{n1}=k$, $\RB_{11}=-2k$ and $\RB_{j 1}=0$ for all $j\neq 2,n$. Thus, for $i\neq i^*$, we obtain:
\begin{align}
\begin{aligned}
\left(\ppFE\right)_i&=\frac{\RB_{i 1}}{\Delta \hvec_{i}+\RB_{11}}=\frac{1}{\beta\theta/k-2}\left(\delta_{2i}+\delta_{ni}\right)\\
(\ppFE)_{i^*=1}&=1-\sum_{i:i\neq i^*} (\ppFE)_i=\frac{\beta\theta/2k-2}{\beta\theta/2k-1} \,.
\end{aligned}
\end{align}
%Note that, for this example, $ \pstB^T\psiM={\theta}/{n}$. Therefore,
The harvesting rate attainable in this regime can be calculated using Eq.~\eqref{eq:Optimium-W-PD}, \begin{align}
\GmaxFE&=\hvec_1+\beta^{-1}\sum_{i:i\neq i^*}\RB_{i 1}\left[\ln (\ppFE)_i-1\right]\nonumber \\
&=\theta-2k{\beta}^{-1}\left[1+\ln\left(\beta \theta k^{-1}-2\right)\right]\,.
\end{align}

\clearpage
\section{Quadratic optimization lemma}
\label{sec:quad-theorem}
% !TeX root = sm.tex

\theoremstyle{plain}
\newtheorem{thm}{\protect\theoremname}
\providecommand{\theoremname}{Lemma}

In this section, we provide a useful theorem for solving the quadratic optimization problems that occurs in our analysis of the linear response (LR) regime. It uses standard techniques from linear algebra.

\newcommand\Lmax{\mathscr{L}^*}
\begin{thm} 
\label{thm:quad}
Consider the maximization problem
\begin{equation}
\Lmax=\max_{\zz\in\mathbb{R}^{n}:\uu^{T}\zz=0}\zz^{T}M\zz+2\zz^{T}\vv,\label{eq:opt-problem2}
\end{equation}
where $M\in\mathbb{R}^{n\times n}$ is a negative semidefinite symmetric matrix
with a single null eigenvector $\uu\in\mathbb{R}^{n}$ and $\vv\in\mathbb{R}^{n}$ is
any vector. 
The solution is given by:
\begin{align}
 & \Lmax=-\vv^{T}M^{+}\vv\ \quad \text{and} \quad  \zz^{*}=-M^{+}\vv\label{eq:specialcase2},
\end{align}
where $M^+$ is the Moore-Penrose pseudo-inverse of $M$.

\end{thm}

\begin{proof}

Since $M$ is symmetric, its real valued eigendecomposition can be written as: $$M= \sum_{\evndx=1}^n\lambda_{\evndx}\uua{\uua}^T$$ where we will choose $\uu=\uuIX{1}$, without loss of generality, i.e. $\lambda_1=0$ is the only zero eigenvalue. The Moore-Penrose pseudo-inverse of $M$ can be written as: $$M^+ = \sum_{\evndx=2}^n \frac{1}{\lambda_\evndx} \uua {\uua}^T.$$

The constraint on the optimization problem reads as $\zz^T\uu=0$, which implies that $$M^+M\zz=(I-\uu \uu^T )\zz=\zz.$$ Now, consider the objective function,
\begin{align}
\zz^{T}M\zz+2\zz^{T}\vv= \zz^{T}M\zz+\vv^{T}M^+M\zz+\zz^{T}MM^+\vv=\left(\zz+M^+\vv\right)^TM\left(\zz+M^+\vv\right)-\vv^T M^+\vv
\end{align}
where we used properties of the transverse operation, symmetry of $M$ and $M^+$, that $M^+M\zz=\zz \Leftrightarrow \zz^T=\zz^TMM^+$, $\vv^T\zz=\zz^T\vv$ and that $M^+MM^+=M^+$. Note that the last term on the right hand side of the previous expression does not depend on $\zz$, and that the first term is always non-positive because $M$ is negative semidefinite. Thus, the objective is maximized only by setting the first term to zero, which is achieved by $\zz=-M^+\vv$. Upon substitution, this yields $\Lmax=-\vv M^+ \vv$ .

\end{proof}

\fi

\end{document}